\documentclass[runningheads,a4paper]{llncs}
\usepackage{epsfig}
\usepackage{amssymb,amsmath}
\usepackage{mathrsfs}
\usepackage{graphicx}
\usepackage{multirow}
\usepackage{color}
\usepackage{url}

\newcommand{\remove}[1]{}

\newcommand{\cM}{{\mathcal{M}}}

\newcommand{\f}{\mathcal{F}}

\newcommand{\rbtl}{\mathcal{F}_{\mathrm{BIT}}^{[n],\rho_r,\rho_w}}  
\newcommand{\bt}{\mathcal{F}_{\mathrm{BIT}}^{[n]}} 
\newcommand{\rbtlr}{\mathcal{F}_{\mathrm{BIT}}^{[n],\rho_r,1}}  
\newcommand{\AO}{\mathcal{F}_{\mathrm{AO}}^{[n],\rho_r,1}} 

\begin{document}
\title{Non-Malleable Codes with Leakage and Applications to Secure Communication }
\author{Fuchun Lin, Reihaneh Safavi-Naini, Mahdi Cheraghchi and Huaxiong Wang}

\maketitle
\begin{abstract} Non-malleable codes are randomized codes 
that protect  coded messages against modification by functions in a tampering  function class. These codes are motivated by providing tamper resilience in applications where 
 a cryptographic secret  is stored in a tamperable storage device and the protection goal is to ensure that the adversary cannot benefit from their  tamperings with the device. In this paper we consider  non-malleable codes for protection of secure communication against active physical layer adversaries. 
We define a class of functions that closely  model tampering of communication by adversaries who can eavesdrop on a constant fraction of the transmitted  codeword, and use this information to select a vector of  tampering functions that will be applied to a second constant fraction of  codeword components (possibly overlapping with the first set).  We derive rate  bounds   for non-malleable codes for this function class and give two modular constructions. The first construction adapts  and provides new analysis for an existing construction in the new setting. The second construction uses a new  approach that results in an explicit construction of non-malleable codes.   We  show applications of our results in securing message communication against active physical layer adversaries in two settings: wiretap II with active adversaries and Secure Message Transmission (SMT) in networks. We discuss our results and directions for future work.

\end{abstract}
\section{Introduction}
Non-Malleable codes  (NM-codes) \cite{DzPiWi} provide protection against active adversaries who can tamper with   coded messages
 using a function from a family $\mathcal{F}$, of tampering functions.
 NM-codes were motivated by providing tamper resilience in  cryptographic applications such as protection of secret keys that are stored in tamperable storage devices (e.g. smart cards) that can be subjected to physical manipulations that  would affect  the values of the  stored secret.
NM-codes ensure the basic security requirement that
the   tampering (using functions from the function class  $\f$) cannot 
 be used to generate related  cryptographic values (e.g. a digital signature for  related keys). 
Roughly speaking, 
a coding scheme  $(\mathsf{Enc}, \mathsf{Dec})$ provides  non-malleability with respect to the tampering family $\mathcal{F}$ if for any  $f\in\mathcal{F}$, a  codeword $\mathbf{c}$ that encodes a message $\mathbf{m}$, the decoding of 
$f( \mathbf{c})$ 
  results in either the original message $\mathbf{m}$,
or a value $\tilde{\mathbf{m}}$ that is unrelated to $\mathbf{m}$, and the probability of which of the two happens is independent of $\mathbf{m}$. 
 This property in the application scenario above will ensure  that the tampering with the device (stored codeword of the key) will result in either an unchanged key, or a key that is unrelated to the original key (and hence an unrelated digital signature). 
 A slightly stronger notion is  {\em strong non-malleability} that  effectively requires that  the   decoded message $\tilde{\mathbf{m}}$   of a modified codeword $f(\mathsf{\mathbf{c}}) =\tilde{\mathbf{c}}$, where $\tilde{\mathbf{c}}  \neq \mathbf{c}$, 
 be independent of  $\mathbf{m}$.  
 %
NM-codes have found other applications in computational cryptography, including construction of non-malleable commitment \cite{Majicrypto,MajiTCC}, and domain extension for public key encryption systems \cite{public key}. NM-codes have been studied in both information-theoretic and computational setting. 
In this paper we consider information-theoretic setting.

%
  Traditional protection goals
against tampering of codewords are error \textit{correction} and error \textit{detection}: correction allows the original message to be recovered, and detection allows the decoder to detect that the message has been modified. 
These protections are  achieved 
  for the class of additive functions    with a bound on the number of  tampered codeword components   (a codeword $\mathbf{c}$ is tampered to $\tilde{\mathbf{c}}=\mathbf{c}\oplus\mathbf{e}$ and $\mathsf{wt}(\mathbf{e})$ is bounded). 
NM-codes  can  provide protection against much more powerful adversaries with access to  much larger function families  
 by using randomised coding schemes  and weakening the protection goal to only ensuring that the adversary cannot benefit  by manipulating a particular message. 
Storage efficiency of non-malleable codes is measured by the {\em rate } of these codes,  given by   the ratio of the message length to the codeword length. The highest achievable rate of  coding schemes for a function family is the {\em capacity }of the coding scheme for the  family.
 
 The ultimate goal of NM-codes is to construct high rate  codes with efficient (computational complexity) encoding and decoding algorithms  for large families of functions. 
In \cite{ChGu0} it is proved that if $|\mathcal{F}|\leq 2^{2^{n\alpha}}$, then the capacity is lower-bounded by 
 $1-\alpha$. The proof uses a probabilistic construction  
of codes that achieves this rate but the code obtained is inefficient ( the construction uses a greedy algorithm). 
Efficient (i.e. polynomial time) non-malleable codes have been constructed \cite{efficientNMC} when the size of the family is 
  $|\mathcal{F}|<2^{p(n)}$ for some polynomial $p$. 
A widely studied family of NM-codes is the bit-wise independent NM-codes that is defined with respect to the bit-wise independent tampering family $\bt$: for binary codewords of length $n$, the tampering function is represented by a vector of $n$ independently chosen functions $(f_1,\cdots,f_n)$, where $f_i$ is a binary   tampering function belonging to $\mathcal{F}_{\mathrm{BIT}}^{[1]} = \{ \mathsf{Set0}, \mathsf{Set1}, \mathsf{Keep}, \mathsf{Flip}\}$, where $\mathsf{Set0}$ and $\mathsf{Set1}$  set the value of the bit to $0$ and $1$, respectively, and  $ \mathsf{Keep}$ and $ \mathsf{Flip}$ will keep and flip the bit, respectively.  
 Non-malleable codes for protection against (simultaneous)  independent bit-wise tampering and
permutation have been proposed in \cite{Majicrypto,MajiTCC}. 
A second widely studied function family  is called {\em $C$-split state model } where for a constant $C$,  the codeword consists of 
$C$ blocks, and each block is tampered independently.  A number of  constructions of these codes  for $C = 2$ in \cite{onebit,additivecombinatorics,NMreduction} and other values of $C$, for example $C = 10$ in \cite{10split},  have been proposed.
All these  function families are naturally suited to the motivating scenario of protecting a stored secret against tampering of the device, and protection approaches that are based on splitting the secret and storing  each part on a different hardware (with the assumption that
they are not all accessible to the adversary). 

In this paper we consider non-malleable codes for protection of communication against physical layer adversaries who tamper with  transmitted codewords. Physical layer security has been pioneered by Wyner \cite{Wyner} who showed message transmission with perfect information-theoretic secrecy and without a shared secret  key is possible  if the adversary does not have full view of the codeword. This incomplete view may be due to the random noise in the adversary's  channel, or  their limited reception and access to the codeword. 
In Wyner's original wiretap model  \cite{Wyner} the eavesdropper's view of the channel is partially obstructed by the noise,  and in Ozarow and Wyner's wiretap II model \cite{WtII}, the eavesdropper can
select the codeword components that they want to eavesdrop, subject to an upper bound on the number of such components.
In both models the adversary is a {\em passive eavesdropping adversary.}
\remove{
Wiretap model has found wide interest in theory and
practice \cite{Matthieu} with thousands of papers in multiple disciplines (e.g. mathematics,
computer science and electrical engineering) on models, foundations and actual implementations. The recent strengthening of security of wiretap codes to the level of cryptographic security (semantic security) \cite{BTV}, has  made these models a strong contender for quantum-safe communication. 

Wyner considered two wiretap models. In the original wiretap model \cite{Wyner} the eavesdropper's view of the channel is partially obscured by the noise. In wiretap II \cite{WtII}, the eavesdropper
selects a subset of components for eavesdropping. In both cases however the adversary is a passive eavesdropping adversary. (The receiver's transmission in the former case is corrupted by the random noise, and in the latter case is noise free.)
}
 Our goal is to  provide protection against active physical layer adversaries that are modelled by a function family.
\subsection*{Our work}
We study NM-codes for a class of  functions that closely reflect capabilities of physical layer adversaries.
We consider adversaries  who have access to directional antennas and advanced  transceivers, and can selectively read (eavesdrop) and tamper with different parts of a codeword. 
The adversary can choose an index set $S_r$ of codeword components to eavesdrop, and an index set $S_w$ of codeword components to tamper with, and the tampering is bitwise (component-wise). The choice of $S_r$ and $S_w$ is adaptive and for each component, taking into account all previous choices that had been made and codeword component values accessed, until that point. The sizes of the two sets are bounded by
$|S_r| \leq n\rho_r$ and $|S_w| \leq n\rho_w$, for two constants $0\leq \rho_r,\rho_w\leq 1$.
We model these adversaries, when the codewords  are binary, by a function family  denoted by $\rbtl$.
The size of this family depends on the actual values of $\rho_r,\rho_w$ and 
is at least $(2^{n\rho_w})^{2^{n\rho_r}}$, which is exponentially larger than the size of $\bt $ 
($4^n$). (This latter class   can be seen as a special case of the former  when $\rho_r=0$ and $\rho_w=1$.)
This is because the eavesdropping set of the adversary allows them to choose their tampering functions depending on the read components of the codeword.
Thus each tampering function $f_i$ will be 
 a function of $\alpha = \mathbf{c}_{S_r}$, that is, the read value of the  codeword $\mathbf{c}$ on the $S_r$ positions.
We obtain rate bounds and give constructions for this class of functions. We also give applications of our results  in two types of communication settings: a wiretap II channel with active adversaries  and secure message transmission in networks.

\vspace{1mm}
\noindent
\textbf{Rate bounds.} 
Storage efficiency of  NM-codes  for a function family 
is measured using the rate of the codes for the family.
We give two sets of results, depending on the non-malleability notion.
For strong non-malleability, 
we  prove capacity of 
 non-malleable codes with respect to $\rbtl$  is 
 $1-\rho_r$. 
 The proof is by deriving an upper bound and a lower bound on the rate of these codes.
  The proof of the upper bound  (Lemma \ref{lem: strong is privacy})  is by proving that    strong non-malleability  with respect to $\rbtl$ implies  indistinguishibility security of the code against an 
  adversary who can eavesdrop  $\rho_r$ fraction of codeword components (wiretap II adversary), and then using rate upper bound of wiretap II codes for this adversary to obtain the upper bound for NM-codes.
The lower bound uses 
 \cite[Theorem 3.1]{ChGu0}  and reduces to finding  an upper bound on the number of functions in $\rbtl$. 

For (default) non-malleability, we prove that if  $\rho_r\leq\rho_w$,  the capacity  of the coding scheme is $1-\rho_r$.
  The rate lower bound follows from the rate lower bound for strong non-malleability,  as it is proved \cite[Theorem 3.1]{DzPiWi} that the latter codes also provide default non-malleability.  To prove the upper bound we 
 build on a result from \cite[Theorem 5.3]{ChGu0} that was proved  for the $C$-split state model. Our proof   requires $\rho_r\leq\rho_w$.
   When $\rho_r> \rho_w$, we show that the rate of NM-codes with respect to $\rbtl$ can exceed $1-\rho_r$.  We leave the upper bound (and hence capacity)  for this case as an  open question.

\vspace{1mm}
\noindent
\textbf{Constructions.} 
We give two constructions.  The first one is based on a modular construction that had been proposed for the function family $\bt$  \cite{DzPiWi}.
We construct a new proof that shows  that with appropriate choice of parameters, one can obtain non-malleability against our new class of tampering functions where the choice of the tampering functions depends on the read values.
The second construction uses a novel approach that relies on a new (not used in the context of NM-codes)  building block and using the security notion of indistinguishability security.

\vspace{1mm}
\noindent
\textit{Construction 1} uses an Algebraic Manipulation Detection (AMD) code \cite{AMD} and a
Linear Error Correcting Secret Sharing (LECSS) \cite{DzPiWi}: the encoding of a message $\mathbf{m}$ is given by LECSS(AMD($\mathbf{m}$)). AMD codes protect against additive errors of oblivious  adversaries (the codeword is not seen by the adversary). A $(t,d)$-LECSS has $t$-uniformity
(every $t$ components is $t$-wise independent, and each bit is uniformly distributed), and  minimum (Hamming) distance $d$. To prove non-malleability, for each  function $f$  we construct a probability distribution 
$\mathcal{D}_f$ that for all  
messages $\mathbf{m}$,  can be used to simulate the decoding of the tampered codeword. 
 The distribution $\mathcal{D}_f$ is obtained by averaging a set of distributions, one for each  read value 
 of the eavesdropped part of the codeword.
 We  borrow techniques from \cite{DzPiWi} and extend them to cater for the new much larger function class.
Theorem \ref{th: basic construction} shows that for $\rbtl$ function class,  to achieve the   level of security  that is provided by a $ (t,d)$-LECSS for the function class $\bt$ (same as $\rho_r=0$ and $\rho_w=1$), 
we need a $(t',d')$-LECSS with $t'=t+n\rho_r$ and $d'=(1-\rho_r)d$. 
That is we need to increase $t$-uniformity of  LECSS to $t' = t+n\rho_r$, but the minimum distance can be reduced.  There is no known construction of LECSS that meets 
the requirements of the construction in \cite{DzPiWi} or our construction,  and so it is unclear if the new set of parameters is harder (or easier) to achieve in concrete constructions.

\vspace{1mm}
\noindent
\textit{Construction 2} uses a novel approach  to the construction of non-malleable codes in
the sense that 
 instead of relying on the $t$-uniformity of LECSS, uses indistinguishability  security of  wiretap II codes.
  The 
  construction uses an AMD code and a linear wiretap II code with indistinguishability security WT: the encoding of a message $\mathbf{m}$ is given by WT(AMD($\mathbf{m}$)). 
  Wiretap II codes are randomised codes that provide indistinguishability security against an eavesdropping adversary that can 
  adaptively eavesdrop a fraction of codeword components. The indistinguishability security is defined as follows:
for 
$|S|\leq n \rho$, 
and 
any two messages $\mathbf{m}_0$ and $\mathbf{m}_1$, 
SD(Enc($\mathbf{m}_0)_{S}$;Enc($\mathbf{m}_1)_{S}$) $\leq \varepsilon$. 
Theorem \ref{con: cute bit'}  shows  
that using a  wiretap II code for   
$\rho=\frac{1+\rho_r}{2}$ and security parameter $\varepsilon$ and an AMD code with error parameter $\delta$, 
results in an NM-code with security parameter $\delta+ 2\epsilon$.
An important advantage of this construction is that there are  explicit constructions for linear 
wiretap II codes 
satisfying $\rho=\frac{1+\rho_r}{2}$ that   use cosets of linear error correcting codes \cite{Markus}, and so we  obtain  an explicit construction of NM-codes with respect to  $\rbtl$ function class using error correcting codes (and using efficient  AMD code construction in \cite{AMD} that has flexible parameters). 
%
A by-product of this construction is an explicit construction of non-malleable codes for $\bt$ 
using error correcting codes. To our knowledge this is the first and the only known direct construction of non-malleable codes for this function family. The code exists for all $n =2^h-1$ and $h\geq 5$.

\vspace{1mm}
\noindent
\textbf{Applications.} 
We motivated the function class  $\rbtl$  by considering physical layer adversaries 
who can  eavesdrop the communication and then choose their tampering functions accordingly. 
The function class also models  adversaries in storage systems \cite{AMD,DzPiWi,ICITS} where the storage is partially leaked to the adversary. 
In the following we apply our results to two  physical layer communication security scenarios that have been widely studied. Before outlining our results,  we  discuss applicability of  \textit{non-malleability as a protection goal in communication security.}

In  the basic physical layer security setting (e.g. wiretap models),  Alice wants to send a message to Bob  and protection is against an eavesdropping adversary.  Using  NM-codes allows protection against  active adversaries with access to a family of tampering functions for which traditional error correction and detection is not possible.  This protection is desirable in cases such as key agreement protocols  where 
the eavesdropper's goal is to influence the shared key.

The protection through NM-code for securing message transmission in the above setting however, does not allow Bob to know if the received message is the  one sent by Alice, or an unrelated one that is the result of tampering. 
\remove{
In this setting Alice and Bob are trusted.

 of secure message transmission  We consider two scenarios: (a) Alice and Bob are trusted and protection is against an outsider adversary; and (b) collusion attack where Alice (or Bob) is dishonest and would like to modify the message without being detected at the point of transmission (reception).
In (a) using non-malleable codes ensures tampering results in a message that is unrelated to the sent one, at the receiver. This is the protection achievable in traditional non-malleable codes. The protection however does not allow the receiver to know if the received message is the sent one, or an unrelated one that is the result of tampering. 
\textcolor{red}{In Section 2 we show that by using structured message spaces, one can also achieve (limited) error detection.} \textcolor{blue}{COMMENT: I think you want to argue a ``AMD style'' result, which says if the offset $\Delta$ is independently added to the codeword, AMD decoder can detect the change. But here it is $\mathbf{m}$ becomes $\tilde{\mathbf{m}}$ and only guarantee is $\tilde{\mathbf{m}}$ is independent of $\mathbf{m}$. Suppose you restrict $\mathbf{m}$ using some structure (known to everyone). The adversary can make $\tilde{\mathbf{m}}$ distributed over messages satisfy that structure. This strategy is valid because $\tilde{\mathbf{m}}$ is independent of the particular message encoded.}
} 
An  interesting application of 
 NM-codes to protection  of  message transmission is against \textit{collusion attacks}, 
 where a dishonest protocol participant  (Alice or Bob) uses a helper to modify the transmitted message to a desired value. 
 Consider 
 a malicious sender who sends a message $\mathbf{m}$, and uses the helper to modify  the codeword during transmission so that the decoded message is a desired value $\mathbf{m}'$. The sender does not have access to an out-of-band channel to send extra information to the helper and the only help they can receive is defined by the class of tampering functions that are available to the helper. Using non-malleable codes with protection against this function family will guarantee that helper cannot help the sender in anyway. An example of such setting is known as \textit{Terrorist Fraud} in authentication protocols \cite{TF}.


\vspace{1mm}
\noindent
\textit{Protecting wiretap II channel against active adversaries.} 
A \textit{$(\rho_r,1)$-active adversary wiretap II code}  is a coding scheme that provides (i) indistinguishability  security against $\rho_r$ leakage,
and (ii) non-malleability against $\rbtlr$ ($\rho_w=1$).   
Our results on 
strong non-malleability can be used to show   that 
 the  secrecy capacity\footnote{ The highest achievable code rate satisfying (i) and (ii).} for $(\rho_r,1)$-active adversary wiretap II code, is $1-\rho_r$ (Theorem \ref{th: active WtII}). 
Our (default) non-malleable code  constructions 
 give 
 (explicit) constructions for 
 $(\rho_r,1)$-active adversary wiretap II codes (Theorem \ref{th: active WtII construction}). 
The  rate of the second construction that uses  wiretap II codes (Section \ref{sec: second construction})
is $\frac{1-\rho_r}{2}$. 
If the wiretap II code in the construction was to provide 
 protection only against eavesdropping, then 
 it could achieve the  rate $1-\rho_r$. Thus  
$\frac{1-\rho_r}{2}$ is the cost of providing non-malleability against tampering family  $\rbtl$ in addition to protection against the eavesdropping adversary. 
A similar construction (different parameters) had been used in  \cite{ICITS} for providing detection of active adversaries in wiretap II setting, when the adversary  uses 
  the eavesdropped part of the codeword to select an {\em offset vector that will be added to the whole codeword},
and it was  proved  that the achievable  rate of the code is $1-\rho_r$.
Their  function family is considerably 
smaller than the family $\rbtlr$, considered here.

\vspace{1mm}
\noindent
\textit{Protecting communication in networks.} Secure Message Transmission (SMT) in networks that are partially controlled by a Byzantine adversary has been studied in \cite{DDWY93}, where the network is modelled as a set of $n$ node disjoint paths (also called {\em wires})  that connect the sender to the receiver. The adversary is active and controls a subset of size $t$ of the wires. 
An 
 $(\varepsilon,\delta)$-SMT protocol guarantees that the information leakage (indistinguishability of adversary's view for two messages) is bounded by $\varepsilon$,  and  
  reliability guarantee  is given by
  $\mathsf{Pr}[M_S \neq M_R]\leq \delta$, where $M_S$ and $M_R$ are the sent and received messages, respectively.  
It has been proved that  
 SMT exists only if $n\geq2t+1$ \cite{Franklin}. 
We define   $(\varepsilon,\delta)$-NM-SMT 
for a network adversary  whose  tampering capability is defined by a function class $\f$,  and
require indistinguishability privacy and reliability   
in terms of non-malleability, against this adversary.
Our construction in Section  \ref{sec: second application} is for $\f$ defined as follows. 
Let $w_i$ denote the   transcript of the $i^{th}$ wire, and let $w_i \in {\cal W}$ for all $i$. The adversary adaptively chooses a set $S\subset [n] $ of   $t$ wires, eavesdrop and arbitrarily   tampers with   them.
The adversary also uses the  values of the eavesdropped wires in $S$  
to tamper with the remaining  $[n]\setminus S$ wires, each by either replacing $w_i$ with   a chosen constant $a_i\in {\cal W}$, or  choosing a constant $a_i\in {\cal W}$ and adding it to $w_i$.
\remove{
construct a value to be added to the transmission or simply block the transmission and replace with a chosen value. 
We call it Additive/Overwrite tampering with $\rho_r=\frac{t}{n}$ leakage. In order to be able to apply our results on non-malleable codes to this new variant of SMT, 
}
The function family defined by the above adversary on ${\cal W}^n \rightarrow {\cal W}^n $ is denoted by $\AO$.
For ${\cal W}= \mathbb{F}_q$ we show that 
the construction  in Section \ref{sec: second construction} 
can be extended to $q$-ary alphabet, resulting in a  $(0,\delta)$-NM-SMT, where $\delta$ is the security parameter of AMD code (Theorem \ref{th: SMT}).

\vspace{1mm}
\noindent
\subsection*{Other related work.}
The concept of non-malleability in cryptography was  introduced by Dolev et al. \cite{NMcryptography} and has since become a fundamental notion in cryptographic systems.
Dziembowski el al. \cite{DzPiWi} introduced non-malleable codes in the context of tamper resilience and providing protection for secrets that are stored in
tamperable hardware.   There is a large body of works on NM-codes including computational NM-codes \cite{Liucomputational,optimalcomputational,Lmoreext},  and  codes with  extra properties such as continuously tampering models \cite{CNMC,TDC,public key},  locally decodable/updatable \cite{LDNMC,ITLDNMC},  and block-wise \cite{block-wise}  that are not directly related to our work. 
In particular, leakage resilient  NM-codes  \cite{Liucomputational} consider a tampering family for non-malleability and a different leakage family for leakage resilience. In our model of NM-codes, there is only one family of functions $\rbtl$ and the goal is non-malleability only. 
Bound on the rate of non-malleable codes was first studied in \cite{ChGu0}.  Authors present a general lower bound for any family of tampering functions that only depends on the size of the family, 
and an upper bound for a 
 family  of tampering functions that arbitrarily act on a subset $S\subset[n]$ of codeword components.  
A tampering class with apparent similarity with our work is   \cite{circuit paper,small-depth circuits}. 
This function class  $\mathsf{Local}^{\ell_o}$ 
 consists of  functions $f:\{0,1\}^n\rightarrow\{0,1\}^n$ 
 where each output bit depends on at most $\ell_o(n)$ input bits.
The tampering functions  in $\rbtl$   are vector of bit tampering functions where  each bit function depends on a subset of size $n\rho_r$ of read components of the vector.
That is unlike the function class  $\mathsf{Local}^{\ell_o}$ where each output bit is determined by a subset of input bits, in  $\rbtl$ a  subset of components of the input codeword determines the vector of functions that will then be applied to the whole codeword.

Non-malleable code constructions for  $\bt$   include,  the first construction  in \cite{DzPiWi}, the first capacity-achieving construction \cite{ChGu1}, capacity-achieving and additionally non-malleable against permutation \cite{MajiTCC}, capacity-achieving and linear time encode/decode \cite{linear time}. Explicit constructions of information-theoretic $C$-split state
include $2$-split state with one-bit message \cite{onebit}, multi-bit message \cite{additivecombinatorics}, constant rate \cite{NMreduction} and $10$-split state \cite{10split}.
Non-malleable codes for non-binary alphabets 
are considered in \cite{linear time} and 
constructions of linear-time encode/decode non-malleable codes  with respect to a tampering class  $\mathcal{F}^+$ that similar to $\bt$ consists of a vector of independently chosen tampering functions, where each function can be  from  $\mathcal{F}$ or an   overwrite functions introduced above (referred to as $ \mathcal{F}_{const}$ in \cite{linear time}), or an identity function $id$.   The code constructions are $\ell$-fold.


\section{Preliminaries}
Coding schemes define the basic properties  for codes that are used in this paper. Let  $\bot$  denote  a special symbol.  

\begin{definition}[\cite{DzPiWi}] A $(k,n)$-coding scheme consists of two functions: a randomised 
encoding function 
$\mathsf{Enc}:\{0,1\}^k\rightarrow\{0,1\}^n$, where the randomness is implicit, and a deterministic decoding function $\mathsf{Dec}:\{0,1\}^n\rightarrow\{0,1\}^k\cup\{\perp\}$ such that, for each $\mathbf{m}\in\{0,1\}^k$, $\mathsf{Pr}[\mathsf{Dec}(\mathsf{Enc}(\mathbf{m}))=\mathbf{m}]=1$ (correctness), and the probability is  over the randomness of the encoding algorithm.
\end{definition}

The rate of a $(k,n)$-coding scheme is the ratio $\frac{k}{n}$.  
For a family of $(k,n(k))$-coding schemes, the  
{\em achievable rates of  the family } is the supremum of the rates of schemes as $k$ grows to infinity.
A {\em tampering function}  
for a $(k,n)$-coding scheme is any function $f:\{0,1\}^n\rightarrow\{0,1\}^n$.

\begin{definition}[\cite{DzPiWi}] \label{def: non-malleability}Let $\mathcal{F}$ be a family of tampering functions.
 For each $f\in\mathcal{F}$ and $\mathbf{m}\in\{0,1\}^k$, define the tampering-experiment
$$
\mathrm{Tamper}_\mathbf{m}^f=\left\{
\begin{array}{c}
\mathbf{x}\leftarrow\mathsf{Enc}(\mathbf{m}),\tilde{\mathbf{x}}= f(\mathbf{x}),\tilde{\mathbf{m}}=\mathsf{Dec}(\tilde{\mathbf{x}})\\
\mathrm{Output}\ \tilde{\mathbf{m}},\\
\end{array}
\right\}.
$$
which is a random variable over the randomness of the encoding function  $\mathsf{Enc}$.
A coding scheme $(\mathsf{Enc},\mathsf{Dec})$ is  {\em non-malleable with respect to $\mathcal{F}$}  if for each $f\in\mathcal{F}$, there exists a distribution
 $\mathcal{D}_f$ over the set $\{0,1\}^k\bigcup\{\perp, \mathsf{same}^*\}$, such that, for all $\mathbf{m}\in\{0,1\}^k$, we have:
\begin{equation}\label{eq: NMCdef}
\mathrm{Tamper}_\mathbf{m}^f\stackrel{\varepsilon}{\approx}\left\{
\begin{array}{c}
\tilde{\mathbf{m}}\leftarrow\mathcal{D}_f\\
\mathrm{Output}\ \mathbf{m}\ \mathrm{ if }\ \tilde{\mathbf{m}}=\mathsf{same}^*,\ \mathrm{and}\ \tilde{\mathbf{m}}\ \mathrm{ otherwise;}
\end{array}
\right\}
\end{equation}
and $\mathcal{D}_f$ is  efficiently samplable given oracle access to $f(\cdot)$. Here ``$\stackrel{\varepsilon}{\approx}$''  refers to statistical or computational indistinguishability. In the case of statistical indistinguishability,   the scheme has exact-security $\varepsilon$  if the above  statistical distance is at most $\varepsilon$.
\end{definition}
The right hand side of (\ref{eq: NMCdef}),   denoted  by $\mbox{Patch}(\mathcal{D}_f,\mathbf{m})$ in \cite{DzPiWi}, is a random variable defined 
by  the distribution $\mathcal{D}_f$ and  the message  $\mathbf{m}$.
Using this notation,  (\ref{eq: NMCdef}) can be written as,
\begin{equation} \label{nmdef}
\mathrm{Tamper}_\mathbf{m}^f\stackrel{\varepsilon}{\approx}\mbox{Patch}(\mathcal{D}_f,\mathbf{m}). 
\end{equation}

A stronger notion of non-malleability is the following.
\begin{definition}[\cite{DzPiWi}] \label{def: strong non-malleability} Let $\mathcal{F}$ be a family of tampering functions. For each $f\in\mathcal{F}$ and $\mathbf{m}\in\{0,1\}^k$, define the tampering-experiment
$$
\mathrm{StrongNM}_\mathbf{m}^f=\left\{
\begin{array}{c}
\mathbf{x}\leftarrow\mathsf{Enc}(\mathbf{m}),\tilde{\mathbf{x}}= f(\mathbf{x}),\tilde{\mathbf{m}}=\mathsf{Dec}(\tilde{\mathbf{x}})\\
\mathrm{Output}\ \mathsf{same}^*\ \mathrm{ if }\ \tilde{\mathbf{x}}=\mathbf{x}\ \mathrm{, and }\ \tilde{\mathbf{m}}\ \mathrm{ otherwise.}\\
\end{array}
\right\},
$$
which is a random variable  over the randomness of the encoding function $\mathsf{Enc}$. A coding scheme $(\mathsf{Enc},\mathsf{Dec})$ is {\em strongly non-malleable w.r.t. $\mathcal{F}$}  if for any $\mathbf{m}_0,\mathbf{m}_1\in\{0,1\}^k$ and any $f\in\mathcal{F}$, we have:
\begin{equation}\label{eq: strongNMCdef}
\mathrm{StrongNM}_{\mathbf{m}_0}^f\stackrel{\varepsilon}{\approx} \mathrm{StrongNM}_{\mathbf{m}_1}^f.
\end{equation}
\end{definition}

It is proved   \cite[Theorem 3.1]{DzPiWi} that strong non-malleability implies (default) non-malleability. The (default) non-malleability however 
is strictly weaker than strong non-malleability and 
does not  imply strong non-malleability.

We will use the following coding schemes in our constructions in Section \ref{sec: constructions}.
\begin{definition}[\cite{AMD}]\label{def: AMD} Let $(\mathsf{AMDenc}, \mathsf{AMDdec})$ be a coding scheme with $\mathsf{AMDenc}: \{0,1\}^k\rightarrow\{0,1\}^n$. We say that $(\mathsf{AMDenc}, \mathsf{AMDdec})$ is a $\delta$-secure Algebraic Manipulation Detection (AMD) code if for all $\mathbf{m}\in\{0,1\}^k$ and all non-zero $\Delta\in\{0,1\}^n$, we have $\mathsf{Pr}[\mathsf{AMDdec}(\mathsf{AMDenc}(\mathbf{m})+\Delta)\notin\{\mathbf{m},\bot\}]\leq \delta$, where the probability is over the randomness of the encoding.
\end{definition}

Efficient AMD codes can be constructed using polynomials over finite fields.
\begin{lemma}[\cite{AMD}]\label{lem: AMD} There exists an AMD code $(\mathsf{AMDenc}, \mathsf{AMDdec})$ with encoder $\mathsf{AMDenc}: \{0,1\}^k\rightarrow\{0,1\}^{k+2u}$ that satisfies $\mathsf{Pr}[\mathsf{AMDdec}(\mathsf{AMDenc}(\mathbf{m})+\Delta)\neq\bot]\leq (k/u+1)/2^u$.
\end{lemma}


Note that the AMD code constructed in Lemma \ref{lem: AMD} is in fact a \textit{tamper detection} code \cite{TDC}, which requires that a tampered codeword is always decoded to $\perp$. 
We say 
an AMD code 
achieves \textit{$\delta$-tamper detection security}  if for all $\Delta\neq 0^n$, $\mathsf{Pr}[\mathsf{AMDdec}(\mathsf{AMDenc}(\mathbf{m})+\Delta)\neq\bot]\leq \delta$.
%

%
The first construction of NM-codes with  respect to $\bt$ in \cite{DzPiWi} uses the following \textit{Linear Error-Correcting Secret Sharing (LECSS)} scheme.
\begin{definition}[\cite{DzPiWi}] \label{def: LECSS} Let $(\mathsf{LECSSenc}, \mathsf{LECSSdec})$ be a coding scheme with messages $\mathbf{m}\in\{0,1\}^k$ and codewords $\mathbf{x}\in\{0,1\}^n$. We say that the scheme is a $(d,t)$-$\mathsf{LECSS}$ if the following properties hold:
\begin{itemize}
\item Linearity: For all $\mathbf{x}\in\{0,1\}^n$ such that $\mathsf{LECSSdec}(\mathbf{x})\neq\bot$, and for all $\mathbf{x}'\in\{0,1\}^n$, we have 
$$
\mathsf{LECSSdec}(\mathbf{x}+\mathbf{x}')=\left\{
\begin{array}{cl}
\bot&,\ \mathrm{if }\ \mathsf{LECSSdec}(\mathbf{x}')=\bot;\\
\mathsf{LECSSdec}(\mathbf{x})+\mathsf{LECSSdec}(\mathbf{x}')&,\ \mathrm{otherwise}.
\end{array}
\right.
$$
\item $d$-distance: For all non-zero $\tilde{\mathbf{x}}\in\{0,1\}^n$ with Hamming weight less than $d$, we have $\mathsf{LECSSdec}(\tilde{\mathbf{x}})=\bot$. 
\item $t$-uniform: For any fixed $\mathbf{m}\in\{0,1\}^k$, we define the random variables $X=(X_1,\cdots,X_n)=\mathsf{LECSSenc}(\mathbf{m})$, where $X_i$ denotes the bit of $X$ in position $i$ and  randomness is from the encoding algorithm. Then the random variables $\{X_i\}_{1\leq i\leq n}$ are individually uniform over $\{0,1\}$ and $t$-wise independent.
\end{itemize}
\end{definition}

In wiretap II model \cite{WtII} 
Alice wants to send messages to Bob over a reliable channel that is eavesdropped by an adversary, 
 Eve, who for a codeword of length $n$,  is allowed to choose any subset of size $\rho n $ of the codeword components for eavesdropping.

\begin{definition}\label{def: WtII} 
A $(\rho,\varepsilon)$-Wiretap II code, or $(\rho,\varepsilon)$-$\mathrm{WT}$ code for short, is a 
$(k,n)$-($q$-ary) coding scheme that satisfies the following privacy property. For any $\mathbf{m}_0,\mathbf{m}_1\in\mathbb{F}_q^k$, any $S\subset [n]$ of size $|S|\leq n\rho$,
\begin{equation}\label{eq: WtII security}
\mathsf{SD}(\mathsf{Enc}(\mathbf{m}_0)_{S};\mathsf{Enc}(\mathbf{m}_1)_{S})\leq \varepsilon.
\end{equation}
A $(\rho,\varepsilon)$-$\mathrm{WT}$ code is called linear if for two vectors $\mathbf{x}_0,\mathbf{x}_1\in\mathbb{F}_q^n$, 
$$
\mathsf{Dec}(\mathbf{x}_0+\mathbf{x}_1)=\left\{
    \begin{array}{cl}
         \bot,&\ \mathrm{ either }\ \mathsf{Dec}(\mathbf{x}_0)=\bot\ \mathrm{ or }\ \mathsf{Dec}(\mathbf{x}_1)=\bot;\\
         \mathsf{Dec}(\mathbf{x}_0)+\mathsf{Dec}(\mathbf{x}_1),&\ \mathrm{otherwise}.
     \end{array}
\right.
$$
\end{definition}

The above indistinguishability based definition of security 
 is  equivalent to 
semantic security which is the strongest notion of security   in cryptography.

\begin{lemma}[\cite{WtII}]\label{lem: wtupperbound}
The rate of wiretap II code with leakage parameter $\rho$ is upper bounded by $1-\rho$.
\end{lemma}
This bound was proved with respect to weak secrecy  \cite{WtII} that assume uniform message distribution and use security measure  $H({\bf M|c_{S}})$, where $H()$ is Shannon entropy and ${\bf M}$ is the random variable associated with the message.
 A $(t, d)$-LECSS construction can be used as a  linear wiretap II code with $\rho=\frac{t}{n}$ and $\varepsilon=0$, however the converse is not true in general. This is because 
privacy requirement of wiretap code is in terms of almost $t$-wise independence instead of $t$-uniformity, and minimum distance of these codes can be $1$.  
Another closely related primitive is {\em  linear secret sharing scheme}, which is  usually 
 studied over large alphabets (share size) and requires reconstruction of message from  subset of codeword components of size above the reconstruction threshold. 

\section{Bit Tampering with Leakage} \label{sec: model and bounds}

Our proposed tampering class is defined by two parameters  $(\rho_r,\rho_w)$. We first define the function class,  and then prove rate bounds for codes that provide non-malleability for this class.  

\subsection{$(\rho_r,\rho_w)_{\mathrm{BIT}}$-$\mathrm{NMC}$}
Let $ \mathcal{F}_{\mathrm{BIT}}^{[1]} = \{ \mathsf{Set0}, \mathsf{Set1}, \mathsf{Keep}, \mathsf{Flip}\}$
denote the set  of functions that tamper with one bit, and $ \mathcal{F}_{\mathrm{BIT}}^{[n]}$ denote the set of $n$-bit bit-wise independent  tampering functions. 
Each $f\in \mathcal{F}^{[n]}_{\mathrm{BIT}}$ is specified by a vector 
$(f_1,f_2,\cdots ,f_n)$ where $f_i\in    \mathcal{F}_{\mathrm{BIT}}^{[1]}$.  
For  a vector $\mathbf{x}=(x_1,x_2,\cdots, x_n)\in \{0,1\}^n$ ,  $f(\mathbf{x})$  is a  vector $\tilde{\mathbf{x}}=(\tilde{x}_1,\tilde{x}_2,\cdots, \tilde{x}_n)\in \{0,1\}^n$, where $\tilde{x}_i= f_i(x_i), i=1,\cdots, n$.  

Let $ [n] =  \{1,2,\cdots n \}$.
We define the set $ \rbtl$ 
as the set of  bitwise tampering functions where the adversary (i)  adaptively selects  a subset $S_r\subset [n]$ of size $n\rho_r$ codeword components for eavesdropping, and (ii)   tampers bitwise with a subset $S_w\subset[n] $ of size $n\rho_w$ of codeword components, each using a function from $ \mathcal{F}_{\mathrm{BIT}}^{[1]} $.  The adversary can choose any pair of subsets $S_r,S_w$, subject to the bound on their sizes.
Let 
${\cal S}^{[n]}_\rho$ 
denote the set of subsets of size $\rho n$ of  $[n]$.
We use ${\cal S}^{[n]}_{\rho_r}$  and ${\cal S}^{[n]}_{\rho_w}$, with cardinality
  $ {n \choose n\rho_r}$
and  ${ n \choose n\rho_w}$, respectively.
For a vector $\mathbf{x}=(x_1,x_2,\cdots, x_n)\in \{0,1\}^n$ and a set $S=\{i_1, i_2\cdots i_{|S|}\}\subset [n]$,  let $\mathbf{x}_{S}$ denote the subvector  $(x_{i_1}, x_{i_2}\cdots x_{i_{|S|}})$.
\begin{definition}\label{def: tampering functions}

Let $\mathbf{x}$  be    a binary vector of length $n$, and   $\mathbf{x }_{S_r}$ and $\mathbf{x}_{S_w}$ denote the subvectors  with components in
the set  ${S_r} \subset [n]$ and ${S_w}\subset [n]$, respectively.  The function  $ g: \{0,1\}^{n\rho_r}\rightarrow   {\cal F}_{\mathrm{BIT}}^{[n\rho_w]}$ defines a vector of  {bit} tampering functions, dependent on  a vector of length $n\rho_r$ (read values in $\mathbf{x }_{S_r}$). 
For fixed  $ {S_r}$,  $ {S_w}$ and $g$ values,   the  tampering function $f_{S_r,S_w,g}:\{0,1\}^n\rightarrow\{0,1\}^n$  takes a vector $\mathbf{x}$ and  results in a vector  $\tilde{\mathbf{x}}=f_{S_r,S_w,g}(\mathbf{x})$ where:
\begin{equation}\label{eq: function}
\tilde{\mathbf{x}}_{S_w} = (g( \mathbf{x}_{S_r}))( \mathbf{x}_{S_w}) \mbox{ and }  \tilde{\mathbf{x}}_{{\bar{S}_w}}=\mathbf{x}_{\bar{S}_w}.    
\end{equation} 
 In other words,    $ f_{S_r,S_w, g}$  modifies
components  of the input vector that are in $S_w$,  {using} the  tampering functions  $g( \mathbf{x}_{S_r})$ and leaves the other components unchanged. 
Let $ g^{\mathbf{x}_{S_r}} \stackrel{def}= g( \mathbf{x}_{S_r})$ denote a vector of $n\rho_w$ {bit} tampering functions.
%
For a fixed  $(\rho_r, \rho_w)$ pair, we define the tampering function family,
\begin{equation}\label{eq: set}
\rbtl    \stackrel{def}=     \left\{f_{S_r,S_w,g}\ | \ S_r\in {\cal S}^{[n]}_{\rho_r} ,S_w\in {\cal S}^{[n]}_{\rho_w},  
g: \{0,1\}^{n\rho_r}\rightarrow   {\cal F}_{\mathrm{BIT}}^{[n\rho_w]}
 \right\}.
\end{equation}

\remove{
where the tampering function $f_{S_r,S_w,g}:\{0,1\}^n\rightarrow\{0,1\}^n$ is defined as follows:
for $\mathbf{x}\in\{0,1\}^n$, depending on the value of $\mathbf{x}_{S_r}$,   $g$ selects a  
bit-wise tampering function  $g( \mathbf{x}_{S_r} )$ from   ${\cal F}_{\mathrm{BIT}}^{ [n\rho_w]} $  
which   is then applied to the  bits in $S_w$.
That is,  the tampering result $f_{S_r,S_w,g}(\mathbf{x})$ for the components in $S_w$ and ${\bar{S}_w}$ are given by,
\begin{equation}\label{eq: function}
(f_{S_r,S_w,g}(\mathbf{x}))_{S_w} = (g( \mathbf{x}_{S_r}))( \mathbf{x}_{S_w}) \mbox{ and }  (f_{S_r,S_w,g}(\mathbf{x}))_{\bar{S}_w}=\mathbf{x}_{\bar{S}_w}.    
\end{equation} 

We write $g^{\mathbf{x}_{S_r}}( \mathbf{x}_{S_w})$  to represent
$(g( \mathbf{x}_{S_r}))( \mathbf{x}_{S_w} )$.

}

\end{definition}

\smallskip
\noindent
\textbf{Sizes of $S_r$ and $S_w$.}  Note that in the above definition,  the sets $S_r$ and $S_w$ have the exact sizes  $n\rho_r$ and $n\rho_w$, respectively.
The  function set  however includes all functions with $|S_r| \leq n\rho_r$ and $|S_w|\leq n\rho_w$. This is because a set $S_r$ of size $ n\rho_r-\ell$, where $\ell$ is an integer satisfying $1\leq \ell <n\rho_r$,  is a subset of  a set $S'_r$ of size $n\rho_r$, where $\ell$ components have not been used in selecting the tampering functions in $g$. Similarly,   a set $S_w$ of size $ n\rho_w-\ell$, where $\ell$ is an integer satisfying $1\leq \ell <n\rho_w$,   is a subset of  a set $S'_w$ of size $n\rho_w$, where $\ell$ components  are $\mathsf{Keep}$ function. 
Thus  although we focus on 
$S_r$ and $S_w$ with exact sizes  $n\rho_r$ and $n\rho_w$,  our results hold if one considers all function vectors that are determined by 
 $S_r$ and $S_w$ of any size up to the corresponding upper bounds. 
This is particularly important as we use this function class for modelling physical layer adversaries, and   adversaries can choose set sizes arbitrarily (up to their reading and writing capabilities).

\smallskip
\noindent
\textbf{Special Case Example $\mathcal{F}_{\mathrm{BIT}}^{[n],0,1}=\bt$.} It is easy to see that in $\rbtl$,  when $\rho_r=0$ the adversary does not have any access to the codeword, and 
$\rho_w=1$ implies that all components of a codeword will be tampered  bitwise and independently. The function class is thus the same as $\bt$.
\remove{
As a quick example, we show $\mathcal{F}_{\mathrm{BIT}}^{[n],\rho_r,\rho_w}$ includes $\mathcal{F}_{\mathrm{BIT}}^{[n]}$ as a special case. 
Let $\rho_r=0$ and $\rho_w=1$.  Here $\rho_r=0$ means adversary cannot see any of the codeword components, and their choice of the tampering function is oblivious.
Because of $\rho_w=1$, the writing set of the adversary is $[n]$ and so the adversary effectively tampers bitwise with the codeword components. This is  
 the set of functions given by $\mathcal{F}_{\mathrm{BIT}}^{[n]}$.
This also follows by noting that
$S_r=\emptyset$ and $S_w=[n]$, 
and $g:\{0,1\}^{0}\rightarrow\mathcal{F}_{\mathrm{BIT}}^{[n]}$ implies that $g$ is a constant function, namely, $g$ is of the form
$g(\cdot)\equiv(f_1,\cdots,f_n)$ for some $(f_1,\cdots,f_n)\in\mathcal{F}_{\mathrm{BIT}}^{[n]}$. 

}

\smallskip
\noindent
\textbf{Subsets of $\rbtl$ for Fixed $(S_r,S_w)$.}
For a fixed pair  
 of reading index set $S_r$ and  writing index set $S_w$,   let, 
\begin{equation}\label{eq: rwset}
\mathcal{F}_{\mathrm{BIT}}^{[n],S_r,S_w}=\left\{f_{S_r,S_w,g}|g:\{0,1\}^{n\rho_r}\rightarrow \mathcal{F}_{\mathrm{BIT}}^{[n\rho_w]} \right \}.
\end{equation}
According to (\ref{eq: set}),  we have: 
\begin{equation}\label{eq: over rwset}
\rbtl=\bigcup_{S_r\in {\cal S}^{[n]}_{\rho_r},S_w\in {\cal S}^{[n]}_{\rho_w}}\mathcal{F}_{\mathrm{BIT}}^{[n],S_r,S_w}.
\end{equation}

\begin{definition}\label{def: rhorrhowNMC} A $(k,n)$-coding scheme is called a  $(\rho_r,\rho_w)_{\mathrm{BIT}}$-Non-Malleable Code ($(\rho_r,\rho_w)_{\mathrm{BIT}}$-$\mathrm{NMC}$)
 if it is a non-malleable coding scheme 
 with respect to $\rbtl$.
\end{definition}

\subsection{Rate Bounds for $(\rho_r,\rho_w)_{\mathrm{BIT}}$-$\mathrm{NMC}$} \label{sec: bounds}
The highest achievable rate of  coding schemes for the function family $\rbtl$  is the {\em capacity } of the coding scheme for this  family.
We provide rate results for the two notions of non-malleability.

{\vspace{1mm}
\noindent
{\bf Strong non-malleability.} This stronger notion puts more stringent requirement on the code and allows us to characterise the capacity for the $\rbtl$ function family.
\begin{theorem}\label{th: strongNM}
The capacity  of strong $(\rho_r,\rho_w)_{\mathrm{BIT}}$-$\mathrm{NMC}$ is $1-\rho_r$. 
\end{theorem}
The proof of this theorem uses a theorem in \cite{ChGu0}  and Lemma \ref{lem: strong is privacy} below, 
to derive a lower bound and an upper bound on the achievable rates of the coding schemes, respectively.  We include the theorem for completeness.
\remove{
 Theorem 3.1
We use two lemmas to prove this theorem. The first lemma (Lemma \ref{th: 1-rho rate}) is due to \cite{ChGu0}. It provides a lower bound on the achievable rate of non-malleable coding schemes for any family $\mathcal{F}$ with a size constraint. In particular, it proves that a probabilistic construction parameterized by $T$ and $\delta$, where $T$ is the number of codewords corresponding to one message and $\delta$ is the relative distance (any two codewords of the coding scheme are $n\delta$ Hamming distance away), gives a non-malleable code with respect to $\mathcal{F}$ with high probability. In the second lemma (Lemma \ref{lem: strong is privacy}), we prove that strong non-malleability implies wiretap II privacy (and hence the upper bound of wiretap II code is an upper bound of $(\rho_r,\rho_w)_{\mathrm{BIT}}$-$\mathrm{NMC}$).
}


\vspace{1mm}
\noindent
{\bf (Theorem 3.1, \cite{ChGu0})}.
{\em 
Let $\mathcal{F}$ be any family of tampering functions from $n$-bit to $n$-bit. 
There exists a construction parameterized by $T$ and $\delta$, such that
for any $\varepsilon,\eta>0$, with probability at least $1-\eta$, the $(k,n)$-coding scheme obtained is a strong non-malleable code with respect to $\mathcal{F}$ with exact security $\varepsilon$ and relative distance $\delta$, provided that both of the following conditions are satisfied. 
\begin{enumerate}
\item $T\geq T_0$, for some
$$
T_0=O\left(\frac{1}{\varepsilon^6}\left(\log\frac{|\mathcal{F}^{[n]}|2^n}{\eta}\right)\right).
$$
\item $k\leq k_0$, for some
$$
k_0\geq n(1-h_2(\delta))-\log T-3\log\left(\frac{1}{\varepsilon}\right)-O(1),
$$
where $h_2(\cdot)$ denotes the binary entropy function. 
\end{enumerate}
Thus by choosing $T=T_0$ and $k=k_0$, the construction satisfies 
$$
k\geq n(1-h_2(\delta))-\log\log\left(\frac{|\mathcal{F}|}{\eta}\right)-\log n-9\log\left(\frac{1}{\varepsilon}\right)-O(1).
$$ 
In particular, if $|\mathcal{F}|\leq 2^{2^{n\alpha}}$ for any constant $\alpha\in(0,1)$, the rate of the code can be made arbitrarily close to $1-h_2(\delta)-\alpha$ while allowing $\varepsilon=2^{-\Omega(n)}$.
}

The following lemma relates strong non-malleability with respect to $\rbtl$, to indistinguishability security of  $(\rho_r,\varepsilon)$-wiretap II codes.

\begin{lemma} \label{lem: strong is privacy} 
If a coding scheme is strongly non-malleable with respect to $\rbtl$ with exact security $\varepsilon$, then it is a $(\rho_r,\varepsilon)$-$\mathrm{WT}$ code.
\end{lemma}
\begin{proof}
Proof is by contradiction: we show that if   a strongly non-malleable coding scheme with respect to $\rbtl$ does not satisfy  wiretap II indistinguishability security,  then we can construct a tampering function that violates the strong non-malleability property of the coding scheme.

Assume  a strongly non-malleable coding scheme with respect to $\rbtl$ does not satisfy  wiretap II indistinguishability security. Then,
there exists a reading set $S_r\subset[n]$ of size $|S_r|=n\rho_r$, and a pair of messages $\mathbf{m}_0,\mathbf{m}_1$ such that 
$$
\mathsf{SD}(\mathsf{Enc}(\mathbf{m}_0)_{S_r};\mathsf{Enc}(\mathbf{m}_1)_{S_r})> \varepsilon.
$$
By the definition of statistical distance, there exists a set $D_\varepsilon\subset\{0,1\}^{n\rho_r}$ such that 
$$
|\mathsf{Pr}[\mathsf{Enc}(\mathbf{m}_0)_{S_r}\in D_\varepsilon]-\mathsf{Pr}[\mathsf{Enc}(\mathbf{m}_1)_{S_r}\in D_\varepsilon]|>\varepsilon.
$$
Now consider a tampering function $f_{S_r,\{1\},g}$,  that reads the  codeword components in $S_r$ positions, and   tampers with the first bit  of the codeword based on the read value. 
We define $g:\{0,1\}^{n\rho_r}\rightarrow\{\mathsf{Set0},\mathsf{Set1},\mathsf{Keep},\mathsf{Flip}\}$ using the set $D_\varepsilon$ as follows.
$$
g(\alpha)=
\left\{
\begin{array}{cl}
\mathsf{Keep},&\alpha\in D_\varepsilon; \\
\mathsf{Flip},&\mbox{otherwise}.\\
\end{array}
\right.
$$
Note that $f_{S_r,\{1\},g}$ when applied to a codeword in  ${\bf c}$, will leave it unchanged if ${\bf c}_{S_r} \in D_\varepsilon$, and flips  its first component otherwise.

According to Definition \ref{def: strong non-malleability}, we should have 
$$
|\mathsf{Pr}[\mathrm{StrongNM}_{\mathbf{m}_0}^{f_{S_r,\{1\},g^{D_\varepsilon}}}=\mathsf{same}^*]-\mathsf{Pr}[\mathrm{StrongNM}_{\mathbf{m}_1}^{f_{S_r,\{1\},g^{D_\varepsilon}}}=\mathsf{same}^*]|>\varepsilon.
$$
This  leads to
$$
\mathsf{SD}(\mathrm{StrongNM}_{\mathbf{m}_0}^{f_{S_r,\{1\},g}};\mathrm{StrongNM}_{\mathbf{m}_1}^{f_{S_r,\{1\},g^{D_\varepsilon}}})> \varepsilon,
$$
which contradicts the strong non-malleability of the coding scheme.
\qed
\end{proof}


We use the above two results leads to  the following proof.
\begin{proof}[of Theorem \ref{th: strongNM}]
Theorem 3.1 in \cite{ChGu0} 
 shows that for any function family $\mathcal{F}$ of size  upper bounded by $|\mathcal{F}|\leq 2^{2^{n\alpha}}$,   there is a family of coding schemes that can  achieve the rate $1-\alpha$ 
 arbitrarily close, by  using sufficiently long codes  
 (e.g. let $\delta=\frac{1}{n}$). 
To use this theorem 
  to find a lower bound on the achievable rate of $\mathcal{F}_{\mathrm{BIT}}^{[n],S_r,S_w}$, we need to upper bound the number of functions in the family.
We note that 
 the representation in (\ref{eq: rwset}) may not be unique for a function in $\mathcal{F}_{\mathrm{BIT}}^{[n],S_r,S_w}$.
 In particular, when $S_r\bigcap S_w\neq \emptyset$, it is possible to have  $f_{S_r,S_w,g'}=f_{S_r,S_w,g''}$ for $g'\neq g''$. 
 For example, let   $S_r=S_w = \{ i \}$. Then  two functions $g'$ and $g''$ from the set $\{g: \{0,1\} \rightarrow \{\mathsf{Set0}, \mathsf{Set1}, \mathsf{Flip}, \mathsf{Keep}\} \}$,
  given by $(g'(x_i=0) = \mathsf{Set1};  g'(x_i=1) = \mathsf{Set0}$, and  $(g''(x_i=0) = \mathsf{Flip};  g''(x_i=1) = \mathsf{Flip})$, 
  will represent the same function in  $\mathcal{F}_{\mathrm{BIT}}^{[n],S_r,S_w}$.
  \remove{
 $i\in S_r\bigcap S_w$ and consider  Then applying $\mathsf{Set0}$ to $x_i$ when $x_i=1$ and $\mathsf{Set1}$ when $x_i=0$ is the same as applying $\mathsf{Flip}$ to $x_i$ for both $x_i=0$ and $x_i=1$. 
But for being able to use
}
 We 
  however only requires an upperbound on the number of functions.
Using   (\ref{eq: over rwset}) it is easy to see that 
 $|\rbtl|\leq{n \choose n\rho_r}{n \choose n\rho_w}(4^{n\rho_w})^{2^{n\rho_r}}$.  
%
\remove{
According to  Theorem 3.1 \cite{ChGu0},
 there exists a strong non-malleable code for any tampering family $\mathcal{F}$ as long as $\log\log |\mathcal{F}|\leq n\alpha$,  and the rate can be made arbitrarily close to $1-\alpha$ for big enough $n$ and when there is no relative distance requirement.
 }
  From the above computation, we have
$$
\begin{array}{ll}
\log\log |\rbtl|&\leq 2(\log n +\log\log n) + n\rho_r + \log (n\rho_w) +1\\
                     &\leq n(\rho_r+\xi),
\end{array}
$$
where $\xi$ is an arbitrarily small constant and the inequality holds for large enough $n$. 
Theorem 3.1 in \cite{ChGu0} shows that 
 for any tampering family $\mathcal{F}$  that satisfies $\log\log |\mathcal{F}|\leq n\alpha$,   there is a   coding scheme with rate arbitrarily close to $1-\alpha$.
 Thus 
 the achievable rate of  $(\rho_r,\rho_w)_{\mathrm{BIT}}$-$\mathrm{NMC}$  with strong non-malleability is lower bounded by $1-\rho_r$. 

The upper bound on the rate of these codes follows from Lemma 
\ref{lem: strong is privacy} that implies  that the rate of a coding scheme with strong non-malleability for function family $\rbtl$, cannot exceed the rate 
of wiretap II codes  of length $n$ and with leakage parameter $\rho_r$, and  noting that the upper bound on the rate of 
these latter codes is 
$1-\rho_r$ (see Lemma \ref{lem: wtupperbound}). 
\qed
\end{proof}

{\vspace{1mm}
\noindent
{\bf Default non-malleability.}
For (default) non-malleability,  we have a general lower bound. But the upper bound (and so capacity)   is only known for $\rho_r\leq\rho_w$.

\begin{theorem} \label{th: defaultNM}
The capacity of $(\rho_r,\rho_w)_{\mathrm{BIT}}$-$\mathrm{NMC}$ for $\rho_r\leq\rho_w$ is $1-\rho_r$.
\end{theorem}

The rate lower bound in the case of (default) non-malleability  follows  from the lower bound on strong non-malleability codes for the same function class,
and noting that a coding scheme that provides strong non-malleability also provides default non-malleability (\cite[Theorem 3.1]{DzPiWi}).
\remove{
$1-\rho_r$ for strong $(\rho_r,\rho_w)_{\mathrm{BIT}}$-$\mathrm{NMC}$ is also a rate lower bound for $(\rho_r,\rho_w)_{\mathrm{BIT}}$-$\mathrm{NMC}$. This is because a strong non-malleable code is always a (default) non-malleable code. 
}
To prove a rate  upper bound for default non-malleability we use the following theorem.



%

\vspace{1mm}
\noindent
{\bf (Theorem 5.3,  \cite{ChGu0})}. 
{\em Let $S\subset[n]$ be of size $\rho n$ and consider the family of  tampering functions that only acts on the coordinate positions in $S$. Then, there is a $\xi_0=O(\frac{\log n}{n})$ such that the following holds. Let $(\mathsf{Enc}, \mathsf{Dec})$ be any $(k,n)$-coding scheme which is non-malleable for the family and achieves rate $1-\rho+\xi$, for any $\xi\in[\xi_0,\rho]$ and error $\varepsilon$. Then $\varepsilon\geq \frac{\xi}{16\rho}$. In particular, when $\rho$ and $\xi$ are absolute constants, $\varepsilon=\Omega(1)$.
}







\begin{proof}[of Theorem \ref{th: defaultNM}]
The lower bound $1-\rho_r$ for strong $(\rho_r,\rho_w)_{\mathrm{BIT}}$-$\mathrm{NMC}$ is also a lower bound for $(\rho_r,\rho_w)_{\mathrm{BIT}}$-$\mathrm{NMC}$. 

In the rest of the proof we show that $1-\rho_r$ is also an upper bound when $S_w\subset S_r$. 
Theorem 5.3 in \cite{ChGu0} 
shows that $1-\frac{|S|}{n}$ is a rate upper bound for non-malleable codes with respect to the family of tampering functions that only act on the coordinate positions in $S\subset [n]$.
We first show  that the set of functions considered in this theorem 
 is the same as the set $\mathcal{F}_{\mathrm{BIT}}^{[n],S, S}$.
Towards this goal, we first show that the size of the two sets are the same. 
%
The total number of   functions that arbitrarily tamper with coordinate positions in $S$ is $(2^{n\rho})^{2^{n\rho}}$.
On the other hand, the set  $\mathcal{F}_{\mathrm{BIT}}^{[n],S,S}$ contains the subset of functions
$$|\left\{f_{S,S,g}|g:\{0,1\}^{n\rho}\rightarrow\{\mathsf{Set0},\mathsf{Set1}\}^{n\rho}\right\}|=(2^{n\rho})^{2^{n\rho}}.$$ 
Note that  
each function in the above description  is distinct  because  tampering of each codeword component in  $S$  can be done in one of the two ways.  
 Thus, $|\mathcal{F}_{\mathrm{BIT}}^{[n],S,S}|\geq (2^{n\rho})^{2^{n\rho}}$.
Noting that the set  $\mathcal{F}_{\mathrm{BIT}}^{[n],S,S}$ is a subset of all functions that tamper with the coordinate positions in $S$, we conclude that  
the two sets have the same size and contain the same functions.
The rate upper bound of $1-\rho_r$ 
 for $(\rho_r,\rho_w)_{\mathrm{BIT}}$-$\mathrm{NMC}$ follows from Theorem 5.3 in \cite{ChGu0} 
 because, when $\rho_r\leq\rho_w$, 
 there exists an  $S\subset [n]$ of size $n\rho_r$ where $\mathcal{F}_{\mathrm{BIT}}^{[n],S,S}\subset\mathcal{F}_{\mathrm{BIT}}^{[n],\rho_r,\rho_w}$,  and for this subset of functions, the upperbound holds.
\qed 
\end{proof}

\begin{remark}
The  proof of Theorem \ref{th: defaultNM} requires  $\rho_r\leq \rho_w$.  
\remove{
above  capacity 
$1-\rho_r$ of (default) $(\rho_r,\rho_w)_{\mathrm{BIT}}$-$\mathrm{NMC}$ holds under the restriction that  
 $\rho_r\leq \rho_w$.  
 }
 For $\rho_r >\rho_w$,  the rate lower bound $1-\rho_r$ remains valid but the upper bound is an open question. 
It is interesting to note that capacity in this case can be higher than $1-\rho_r$. This is  because for small values of 
$\rho_w$,  error correcting codes with  non-zero rate exists and in the case of  error correcting codes  $\rho_r=1$ which using 
$1-\rho_r=0$,  suggests zero rate for NM-codes.  This is however not true 
because error correcting codes are non-malleable and in this case have  non-zero rate.
\end{remark}

\smallskip
\section{Code Constructions}\label{sec: constructions}
Using the results in 
\cite{TDC}, one can  construct NM-codes for 
 tampering family 
 (including  $\rbtl$ family) 
  in the Common Reference String (CRS) model.  
We construct explicit 
and efficient $(\rho_r,\rho_w)_{\mathrm{BIT}}$-$\mathrm{NMC}$ without any setup conditions.   

Our first construction is based on a construction proposed by Dziembowski, Pietrzak and Wichs \cite{DzPiWi} for the set of Bit-wise Independent Tampering (BIT) functions ($\mathcal{F}_{\mathrm{BIT}}^{[n]}$ in our notation).
This construction has inspired a number of other NM-code constructions  \cite{Majicrypto}, \cite{ChGu1}, {\cite{circuit paper} and more recently \cite{linear time}, \cite{public key}. 
The construction 
uses two coding schemes: an AMD (Algebraic  Manipulation Detection, see Definition \ref{def: AMD}) code and a LECSS (Linear Error Correcting Secret Sharing, see Definition \ref{def: LECSS}) with appropriate parameters. 
Explicit construction of LECSS with the required parameters has been an open question. 
%
Our second construction uses 
a linear wiretap II code  and an AMD code.
\remove{
Wiretap II codes are easier to construct. With this construction, we are able to obtain explicit $(\rho_r,\rho_w)_{\mathrm{BIT}}$-$\mathrm{NMC}$s. For example, we obtain a family of $(0,1)_{\mathrm{BIT}}$-$\mathrm{NMC}$s (this is just the BIT non-malleable codes) that can be described using Hamming codes and arithmetic over finite field with characteristic $2$.
}

%


\subsection{Construction 1: LECSS$\circ$AMD }\label{sec: basic construction}
%
We consider the function
class $\rbtl$ with size at least 
$(2^{n\rho_w})^{2^{n\rho_r}}$, which is much larger than $\bt$ (of size $4^n$) that was considered in 
 \cite{DzPiWi}. 

\begin{theorem}\label{th: basic construction} Let $(\mathsf{LECSSenc}, \mathsf{LECSSdec})$ be a $(d',t')$-$\mathrm{LECSS}$ with an encoder $\mathsf{LECSSenc}: \{0,1\}^{\ell}\rightarrow\{0,1\}^n$. Let $(\mathsf{AMDenc}, \mathsf{AMDdec})$ be an AMD code from $\{0,1\}^k$ to $\{0,1\}^\ell$ with $\delta$-tamper detection security. Let $(\mathsf{Enc}, \mathsf{Dec})$ be defined as follows.
\begin{equation}\label{eq: construction}
\left\{
\begin{array}{ll}
\mathsf{Enc}(\mathbf{m})&=\mathsf{LECSSenc}(\mathsf{AMDenc}(\mathbf{m}))\\
\mathsf{Dec}(\mathbf{x})&=\mathsf{AMDdec}(\mathsf{LECSSdec}(\mathbf{x}))\\
\end{array}
\right.
\end{equation}
Then the $(k,n)$-coding scheme $(\mathsf{Enc}, \mathsf{Dec})$ is a $(\rho_r,1)_{\mathrm{BIT}}$-$\mathrm{NMC}$ with exact security 
$\max\{\delta,2^{-\Omega(t'- n\rho_r)}\}$,
if $t'> n\rho_r$ and $d'>\frac{n(1-\rho_r)}{4}$.
\remove{
if depending on the $(\rho_r,\rho_w)$ pair, the parameter of $(d',t')$-$\mathrm{LECSS}$ satisfy the following.
\begin{enumerate}
\item For $\rho_w\geq\frac{1-\rho_r}{2}$: $(d',t')$-$\mathrm{LECSS}$ satisfies $t'> n\rho_r$ and $d'>\frac{n(1-\rho_r)}{4}$;\\
\item For $\rho_w<\frac{1-\rho_r}{2}$: $(d',t')$-$\mathrm{LECSS}$ satisfies $t'> n\rho_r$ and $d'>\frac{n\rho_w}{2}$.\\
\end{enumerate}
}
\end{theorem}

Compared to Theorem 4.1 in \cite{DzPiWi}, the result implies that for the same security level, one needs to use LECSS with higher uniformity parameter ($t' = t+n\rho_r$) but the minimum distance of the LECSS can be somewhat relaxed.}
\remove{
The proof of the theorem uses  two 
 lemmas. 
 Lemma \ref{lem: rhow=1}   gives the parameters of  LECSS when 
 $\rho_w=1$.  We noted earlier that protection against $\rbtl$ implies protection for all $S_r$ and $S_w$ with sizes up to   $n\rho_r $ and $ n\rho_w$, respectively, and so 
 $(\rho_r,1)_{\mathrm{BIT}}$-$\mathrm{NMC}$ is also a $(\rho_r,\rho_w)_{\mathrm{BIT}}$-$\mathrm{NMC}$ for any $\rho_w<1 $. 
 Lemma \ref{lem: small rhow} shows  how the parameters of the $(d',t ')$-$\mathrm{LECSS}$, in particular  
 $d'$, can be relaxed when $\rho_w$ is less than $\frac{1-\rho_r}{2}$. 
}
%
The intuition of the proof is  as follows. 
We need to show that for an arbitrary function $f_{S_r,[n],g}$ there is a distribution ${\cal D}^{f_{S_r,[n],g}}$ that satisfies (\ref{nmdef}) for any message $\mathbf{m}$.
For a function with  read index set $S_r$, the set of codewords corresponding to the message $\mathbf{m}$ can be partitioned into subsets $C_{\alpha} $ consisting of codewords $\mathbf{c}$ where $\mathbf{c}_{S_r} =\alpha$ and  $\alpha\in\{0,1\}^{n\rho_r}$.
For  all codewords in $C_{\alpha} $ the   function $  g^\alpha\stackrel{def}{=}g(\alpha)  \in \mathcal{F}_{\mathrm{BIT}}^{[n] }$ will be used.
 By choosing appropriate parameters for LECSS and AMD code, we can construct distribution ${\cal D}^{f_{S_r,[n],g}}$ which is the ``average" of the distributions ${\cal D}^{f_{S_r,[n],g}}_\alpha$ corresponding to $C_{\alpha} $.
 \remove{
 To construct ${\cal D}^{f_{S_r,[n],g}}_\alpha$ we will use an approach similar to \cite{DzPiWi}.
 Our results show that to maintain the  same level of security while allowing tampering to depend on $\rho_r n$ read components, 
  the parameters $t'$ and $d'$ of LECSS must be chosen as  $t' = t + n\rho_r$ and  $d' = (1-\rho_r) d$.
}

\begin{proof}

Consider a message $\mathbf{m}\in\{0,1\}^k$, and a tampering function $f_{S_r,[n],g}\in\rbtlr$.  We define two (vector) random variables $\mathbf{X},\tilde{\mathbf{X}}\in\{0,1\}^n$ representing the codeword and the tampered codeword, respectively. 
$$
\mathbf{m}\stackrel{\mathsf{Enc}}{\longrightarrow} \mathbf{X}\stackrel{f_{S_r,[n],g}}{\longrightarrow} \tilde{\mathbf{X}}\stackrel{\mathsf{Dec}}{\longrightarrow} \mathrm{Tamper}_{\mathbf{m}}^{f_{S_r,[n],g}}.
$$
The randomness of the variables
$\mathrm{Tamper}_{\mathbf{m}}^{f_{S_r,[n],g}}$, $\mathbf{X}$ and $\tilde{\mathbf{X}}$  are from the randomness of the encoding.
Since $\mathrm{Tamper}_{\mathbf{m}}^{f_{S_r,[n],g}}$ and $\mathbf{X}$ are correlated, we  have

\begin{eqnarray} \label{sum}
&&\mathsf{Pr}\left[\mathrm{Tamper}_{\mathbf{m}}^{f_{S_r,[n],g}}=\gamma\right]\nonumber\\
&&  \hspace{0.8cm}= \sum_{\alpha\in\{0,1\}^{n\rho_r}} \mathsf{Pr}[\mathbf{X}_{S_r}=\alpha]\cdot\mathsf{Pr}[\mathrm{Tamper}_{\mathbf{m}}^{f_{S_r,[n],g}}=\gamma|\mathbf{X}_{S_r}=\alpha] \\
&&    \hspace{0.8cm}=  \sum_{\alpha\in\{0,1\}^{n\rho_r}} \frac{1}{2^{n\rho_r}}\mathsf{Pr}[\mathrm{Tamper}_{\mathbf{m}}^{f_{S_r,[n],g}}=\gamma|\mathbf{X}_{S_r}=\alpha], \nonumber
  \end{eqnarray}
where the last equality follows   from the $t'$-uniform property of the LECSS, and assuming that $t' > n\rho_r$ .

To construct the distribution ${\cal D}^{f_{S_r,[n],g}}$ that satisfies (\ref{nmdef}), we  start by constructing a set of distributions $\{{\cal D}^{f_{S_r,[n],g}}_\alpha|\alpha\in\{0,1\}^{n\rho_r}\}$, each  over the set $\{0,1\}^k\bigcup\{\bot\}\bigcup\{\mathsf{same}^{*}\}$ and satisfying
\begin{equation}\label{eq: conditioned}
\mathrm{Tamper}_{\mathbf{m}, 
\alpha}^{f_{S_r,[n],g}}   \stackrel{def}{=}
\left(\mathrm{Tamper}_{\mathbf{m}}^{f_{S_r,[n],g}}|\mathbf{X}_{S_r}=\alpha\right) \stackrel{\varepsilon}{\approx}\mbox{Patch}({\cal D}^{f_{S_r,[n],g}}_\alpha,\mathbf{m}),
\end{equation}
for $\varepsilon=\max\{\delta,2^{-\Omega(t'- n\rho_r)}\}$.
 The distribution ${\cal D}^{f_{S_r,[n],g}}_\alpha$  is used to simulate 
 the function  $f_{S_r,[n],g}$  when 
  applied to codewords in $C_\alpha$, the set of encodings  ${\bf c}$ of 
  $\mathbf{m}$ that for  the  chosen index set $S_r$, have
 ${\bf c}_{S_r}=  \alpha$.  
 From (\ref{sum}) we have,
\begin{eqnarray*}
\mbox{SD}\left(\mathrm{Tamper}_{\mathbf{m}}^{f_{S_r,[n],g}};  {\cal D}^{f_{S_r,[n],g}}\right) \leq 2^{-n\rho_r} \sum_\alpha \mbox{SD}\left(\mathrm{Tamper}_{\mathbf{m}, 
\alpha}^{f_{S_r,[n],g}} , {\cal D}^{f_{S_r,[n], g}}_\alpha\right),
\end{eqnarray*}
where $\mathrm{Tamper}_{\mathbf{m}, \alpha}^{f_{S_r,[n],g}}$ is the tampering variable defined in (\ref{eq: conditioned}). 
We will have
\begin{eqnarray*}
\mbox{SD}\left(\mathrm{Tamper}_{\mathbf{m}}^{f_{S_r,[n],g}},  {\cal D}^{f_{S_r,[n],g}}\right) \leq \epsilon \end{eqnarray*}
because  for all $\alpha$ we have $ \mbox{SD}\left(\mathrm{Tamper}_{\mathbf{m}, 
\alpha}^{f_{S_r,[n],g}} , {\cal D}^{f_{S_r,[n], g}}_\alpha\right) \leq \epsilon$.

To construct ${\cal D}^{f_{S_r,[n],g}}_\alpha$, 
consider  $\mathbf{X} \in C_\alpha $, namely, assume $\mathbf{X}_{S_r}=\alpha$.  
Let $g(\alpha)=g^\alpha=(g^\alpha_1,\cdots,g^\alpha_n)$, where $g^\alpha_i \in \{ \mathsf{Set0}, \mathsf{Set1}, \mathsf{Keep}, \mathsf{Flip}\}$.
 
 Firstly,  
 $\left(g^\alpha(\mathbf{X})\right)_{S_r}$  on condition $\mathbf{X}_{S_r}=\alpha$  
 will be constant 
and can be computed from $\{g^\alpha_i|i\in S_r\}$ and $\alpha$. 

Next consider application of    $\{g^\alpha_i|i\in \bar{S}_r\}$ to 
 $\mathbf{X}_{\bar{S}_r}$. The analysis below is all under the condition $\mathbf{X}_{S_r}=\alpha$.
 For all 
component functions in $\{g^\alpha_i|i\in \bar{S}_r\}$ that are in  $\{\mathsf{Set0}, \mathsf{Set1} \}$ 
the values of $g^\alpha(\mathbf{X})$ at these positions will be 
 constant values $0$ and $1$, respectively. 
For all 
component functions in $\{g^\alpha_i|i\in \bar{S}_r\}$ that are in  $\{\mathsf{Keep}, \mathsf{Flip} \}$ 
 the values of $g^\alpha(\mathbf{X})$ in these positions will be kept the same and flipped, respectively.
 In the latter case this means that the statistical  properties of columns  ($C_\alpha$ seen as an array of row vectors) will stay the same.
 Since columns of $C_\alpha$ in $\bar{S}_r$  are $(t'-n\rho_r)$-wise independent, we will have, (i) each {non-overwritten} column of $g^\alpha(C_\alpha)$ (also as an array of row vectors)  in $\bar{S}_r$ is uniformly distributed, and (ii) {non-overwritten} columns of $g^\alpha(C_\alpha)$  in $\bar{S}_r$ are jointly $(t'-n\rho_r)$-wise independent.

Let $n^{ow}_{\bar{S}_r}$ 
denote the number of overwrite bit functions in $\{g^\alpha_i|i\in \bar{S}_r\}$ defined as,
$$
n^{ow}_{\bar{S}_r}=|\{i\in\bar{S}_r|g^\alpha_i=\mathsf{Set0} \mbox{ or } g^\alpha_i=\mathsf{Set1}\}|.
$$
The above analysis shows that $n\rho_r+n^{ow}_{\bar{S}_r}$ components of $g^\alpha(\mathbf{X})$  will have constant values 
independent of the initial value of $ \mathbf{X}$, 
 while the remaining $(|\bar{S}_r|-n^{ow}_{\bar{S}_r})$ non-overwritten components in $\bar{S}_r$ are individually  uniformly distributed,
  and are jointly $(t'-n\rho_r)$-wise independent.
  
For $g^\alpha\in\bt$, the   difference function  $\Delta g^\alpha\in\bt$  is defined as:
\begin{equation}\label{eq: difference}
\Delta g^\alpha(\mathbf{x})=g^\alpha(\mathbf{x})\oplus\mathbf{x}.
\end{equation}
\remove{We observe that  applying the difference function on vectors in $C_\alpha$   will result in columns that will be overwritten by 
 $g^\alpha$,
with columns that will become  non-overwritten by the difference function, 
  and vice versa.} 
Using  (\ref{eq: difference}) it can be seen that
 if $g^\alpha_i\in\{\mathsf{Set0},\mathsf{Set1}\}$, then $\Delta g^\alpha_i\in\{\mathsf{Keep},\mathsf{Flip}\}$, and  if $g^\alpha_i\in\{\mathsf{Keep},\mathsf{Flip}\}$, then $\Delta g^\alpha_i\in\{\mathsf{Set0},\mathsf{Set1}\}$.
Thus applying a non-overwrite bit function in $\{\Delta g^\alpha_i|i\in \bar{S}_r\}$ on a column of   $C_\alpha$,  will  correspond to applying
an overwrite function of   $\{g^\alpha_i|i\in \bar{S}_r\}$ on that column, and vice versa.

The distribution $\mathcal{D}^{f_{S_r,[n],g}}_\alpha$ is constructed by considering   four cases according to the number  of overwrite component functions, denoted by $n^{ow}_{\bar{S}_r}$,   in the  set ${\bar{S}_r}$. In the following analysis, following the approach of \cite[Appendix B]{DzPiWi}, 
we consider four cases.

\begin{enumerate}
\item $n^{ow}_{\bar{S}_r}\in \left[0,t'-n\rho_r\right]$: rely on linearity, $t'$-uniform of LECSS and AMD;
\item $n^{ow}_{\bar{S}_r}\in \left(t'-n\rho_r,\frac{|\bar{S}_r|}{2}\right]$: rely on linearity, $t'$-uniform and $d'$-distance of LECSS;
\item $n^{ow}_{\bar{S}_r}\in \left(\frac{|\bar{S}_r|}{2}, n-t'\right)$: rely on $t'$-uniform and $d'$-distance of LECSS;
\item $n^{ow}_{\bar{S}_r}\in \left[n-t',|\bar{S}_r|\right]$: rely on $t'$-uniform of LECSS.
\end{enumerate}
For each case we show how  the distribution   can be constructed. 
The complete proof is   given in Appendix \ref{apdx: proof of basic construction}.

\qed
\end{proof}

\begin{lemma}\label{lem: small rhow}
When $\rho_w<\frac{1-\rho_r}{2}$, the coding scheme $(\mathsf{Enc}, \mathsf{Dec})$ in Theorem \ref{th: basic construction} is a $(\rho_r,\rho_w)_{\mathrm{BIT}}$-$\mathrm{NMC}$ with exact security 
$\max\{\delta,2^{-\Omega(t'- n\rho_r)}\}$,
if the $(d',t')$-$\mathrm{LECSS}$ satisfies $t'> n\rho_r$ and $d'>\frac{n\rho_w}{2}$.
\end{lemma}

\begin{proof}
Lemma \ref{lem: small rhow} is a special case of Theorem \ref{th: basic construction} when $\rho_w$ is small and we have $\rho_w<\frac{1-\rho_r}{2}$.
Using the proof steps of this theorem leads to 
four cases that are distinguished according to $n^{ow}_{\bar{S}_r}$,   the number of overwrite component functions of $g^\alpha$ in $\bar{S}_r$. Note that we always have $n^{ow}_{\bar{S}_r}\leq n\rho_w$ because $n\rho_w$ is the total writing budget. 
When $\rho_w<\frac{1-\rho_r}{2}$, we have $n^{ow}_{\bar{S}_r}\leq n\rho_w<\frac{|\bar{S}_r|}{2}$ and hence Case 3 and Case 4 in the proof above will not occur. 
If Case 2. occurs (i.e. when $n\rho_w>t'-n\rho_r$), the range of $n^{ow}_{\bar{S}_r}$ is $(t'-n\rho_r,n\rho_w]\subsetneq (t'-n\rho_r,\frac{|\bar{S}_r|}{2}]$  and this leads to the relaxation of the parameter $d'$ from $d'>\frac{|\bar{S}_r|}{4}$ to $d'>\frac{n\rho_w}{2}$. 
The rest of the argument is as before, and 
is given in the detailed arguments of Case 2,  given in Appendix \ref{apdx: proof of basic construction}.
\qed
\end{proof}

\remove{
\begin{proof}
The intuition of Lemma \ref{lem: small rhow} is, in the final steps of the proof above, we consider four cases according to the number $n^{ow}_{\bar{S}_r}$ of overwrite component functions of $g^\alpha$ in $\bar{S}_r$. Note that we always have $n^{ow}_{\bar{S}_r}\leq n\rho_w$ because $n\rho_w$ is the total writing budget. When $\rho_w<\frac{1-\rho_r}{2}$, we have $n^{ow}_{\bar{S}_r}\leq n\rho_w<\frac{|\bar{S}_r|}{2}$ and hence Case 3. and Case 4. in the proof above will not occur. Also if Case 2. occurs (i.e. when $n\rho_w>t'-n\rho_r$), the range of $n^{ow}_{\bar{S}_r}$ is $(t'-n\rho_r,n\rho_w]\subsetneq (t'-n\rho_r,\frac{|\bar{S}_r|}{2}]$. This leads to the relaxation on the parameter $d'$ from $d'>\frac{|\bar{S}_r|}{4}$ to $d'>\frac{n\rho_w}{2}$. This step follows from the detailed arguments of Case 2. given in Appendix \ref{apdx: proof of basic construction}.
\qed
\end{proof}
}

Explicit construction of LECSS that satisfies  the required minimum distance and uniformity for arbitrary security level is an open question.  The probabilistic construction in \cite[Lemma C.2]{DzPiWi} could be used to estimate 
 the achievable rate of LECSS   with the required parameters. The estimate for the original parameters of LECSS shows positive achievable rate \cite[Theorem 4.2]{DzPiWi}.  The second construction uses building  blocks for which explicit constructions do exist. However estimating  achievable rate of these codes remain open.

\subsection{Construction 2: WT$\circ$AMD  }\label{sec: second construction} 

This is a modular construction that uses a wiretap II code and an AMD code with appropriate parameters.

 
\begin{theorem}\label{con: cute bit'}
Let $(\mathsf{AMDenc}, \mathsf{AMDdec})$ be an AMD code from $\{0,1\}^{k}$ to $\{0,1\}^{\ell}$ with  $\delta$-tamper detection security.  
Let $(\mathsf{WTenc}, \mathsf{WTdec})$ be a linear $(\rho,\varepsilon)$-$\mathrm{WT}$ code with encoder $\mathsf{WTenc}:\{0,1\}^{\ell}\rightarrow\{0,1\}^n$.   Let $(\mathsf{Enc}, \mathsf{Dec})$ be defined as follows.
\begin{equation}\label{eq: construction cuit bit'}
\left\{
\begin{array}{ll}
\mathsf{Enc}(\mathbf{m})&=\mathsf{WTenc}(\mathsf{AMDenc}(\mathbf{m}));\\
\mathsf{Dec}(\mathbf{x})&=\mathsf{AMDdec}(\mathsf{WTdec}(\mathbf{x})).\\
\end{array}
\right.
\end{equation}
Then the $(k,n)$-coding scheme $(\mathsf{Enc}, \mathsf{Dec})$ is a $(\rho_r,1)_{\mathrm{BIT}}$-$\mathrm{NMC}$ with exact security $2\varepsilon+\delta$, 
if $\rho\geq\frac{1+\rho_r}{2}$.
\remove{
the leakage  parameter of  the $(\rho,\varepsilon)$-$\mathrm{WT}$ satisfy the following:
\begin{enumerate}
\item 
$\rho\geq\frac{1+\rho_r}{2}$,  if $\rho_w\geq\frac{1-\rho_r}{2}$; \\
\item 
$\rho\geq\rho_r+\rho_w$,  if  $\rho_w<\frac{1-\rho_r}{2}.$
\end{enumerate}
}
The rate of the NM-code is  upper bounded by $1-\rho$. 
\end{theorem}

\remove{
The proof of this theorem, similar to  Theorem \ref{th: basic construction},   has two steps. First, proving the result for $(\rho_r,1)_{\mathrm{BIT}}$-$\mathrm{NMC}$  (Lemma \ref{lem: rhow=1'}), and noting that $\rbtl$ includes all functions with $|S_w| \leq \rho_w n$,   conclude the result for $(\rho_r,\rho_w)_{\mathrm{BIT}}$-$\mathrm{NMC}$.  Then prove a second lemma (Lemma \ref{lem: small rhow'}) for the special case that $\rho_w \leq \frac{1-\rho_r}{2}$.
}

\remove{

We use two lemmas to prove Theorem \ref{con: cute bit'}. Lemma \ref{lem: rhow=1'} below gives the parameter requirement of the $(\rho,\varepsilon)$-$\mathrm{WT}$ for the extreme case $\rho_w=1$. Note that a 
$(\rho_r,1)_{\mathrm{BIT}}$-$\mathrm{NMC}$ is a $(\rho_r,\rho_w)_{\mathrm{BIT}}$-$\mathrm{NMC}$ for any $\rho_w$. We then prove in the second lemma (Lemma \ref{lem: small rhow'}) how the parameter requirement of $\rho$ can be relaxed when $\rho_w$ is less than $\frac{1-\rho_r}{2}$. 
}

\remove{
\begin{lemma}\label{lem: rhow=1'}
The coding scheme $(\mathsf{Enc}, \mathsf{Dec})$ in Theorem \ref{con: cute bit'} is a $(\rho_r,1)_{\mathrm{BIT}}$-$\mathrm{NMC}$ with exact security 
$2\varepsilon+\delta$,
if the $(\rho,\varepsilon)$-$\mathrm{WT}$ satisfies $\rho\geq\frac{1+\rho_r}{2}$.
\end{lemma}
}

\begin{proof}
We use the  approach 
 of Theorem \ref{th: basic construction} and express
$\mathrm{Tamper}_{\mathbf{m}}^{f_{S_r,[n],g}}$ as  in expression (\ref{sum}). We thus need to find 
distribution ${\cal D}^{f_{S_r,[n],g}}_\alpha$ (independent of the message $\mathbf{m}$) that  corresponds to each 
$$
\mathrm{Tamper}_{\mathbf{m}, \alpha}^{f_{S_r,[n],g}}   =
\left(\mathrm{Tamper}_{\mathbf{m}}^{f_{S_r,[n],g}}|\mathbf{X}_{S_r}=\alpha\right).
$$

The proof however,  replaces $t$-uniformity in LECSS with  $t$-privacy in  wiretap II code which is expressed in terms of indistinguishability security.
We use this property and the linearity of the code, to show that the distributions ${\cal D}^{f_{S_r,[n],g}}_\alpha$  can be found using wiretap II encodings of $0^\ell$. Define a (vector) random variable $\mathbf{Y}=\mathsf{WTenc}(0^\ell)$ and then construct the distribution ${\cal D}^{f_{S_r,[n],g}}$ as follows:
$$
\mathsf{Pr}\left[{\cal D}^{f_{S_r,[n],g}}=\gamma\right]\nonumber= \sum_{\alpha\in\{0,1\}^{n\rho_r}} \mathsf{Pr}[\mathbf{Y}_{S_r}=\alpha]\cdot\mathsf{Pr}[{\cal D}^{f_{S_r,[n],g}}_\alpha=\gamma].
$$

We use the following notations. Let $\mathbf{X}=\mathsf{Enc}(\mathbf{m})$ and 
 $\mathbf{X}^\alpha=(\mathbf{X}|\mathbf{X}_{S_r}=\alpha)$. The tampered version is given by $\tilde{\mathbf{X}}^\alpha=f_{S_r,[n],g}(\mathbf{X}^\alpha)=g^\alpha(\mathbf{X}^\alpha)$, where $g^\alpha=g(\alpha)\in\bt$ and $g^\alpha=(g^\alpha_1,\cdots,g^\alpha_n)$. 
Since $\mathbf{X}^\alpha_{S_r}=\alpha$,   then  $\tilde{\mathbf{X}}^\alpha_{S_r}$ is a constant. 
Components of $\tilde{\mathbf{X}}^\alpha_{\bar{S}_r}$ can be constant or random bits, depending on the corresponding components $g^\alpha_i,i\in\bar{S}_r$. We consider 
 two cases that are distinguished by  the number $n^{ow}_{\bar{S}_r}$ of overwrite bit functions in $\{g^\alpha_i|i\in \bar{S}_r\}$.

\begin{enumerate}
\item ``At most half of $\bar{S}_r$ are overwrite functions'' ($n^{ow}_{\bar{S}_r}\leq \frac{|\bar{S}_r|}{2}$): 

The difference function $\Delta g^\alpha$ (i.e. $g^\alpha(\mathbf{x})=\mathbf{x}\oplus\Delta g^\alpha(\mathbf{x})$) has at most half non-overwrite bit functions over $\bar{S}_r$. Let $S$ be the index set of the non-overwrite components of $\Delta g^\alpha$ in $\bar{S}_r$. Then $|S|\leq\frac{|\bar{S}_r|}{2}=\frac{n(1-\rho_r)}{2}$ and hence $\frac{|S_r\bigcup S|}{n}\leq \frac{1+\rho_r}{2}$. We use the short hand $\mathbf{Y}^\alpha=(\mathbf{Y}|\mathbf{Y}_{S_r}=\alpha)$ (similar to $\mathbf{X}^\alpha=(\mathbf{X}|\mathbf{X}_{S_r}=\alpha)$). Then according to the indistinguishability privacy of the $(\frac{1+\rho_r}{2},\varepsilon)$-$\mathrm{WT}$ code, we have 
\begin{equation}\label{eq: ssr}
\mbox{SD}(\mathbf{X}^\alpha_S;\mathbf{Y}^\alpha_S)\leq\mbox{SD}(\mathbf{X}_{S_r\bigcup S};\mathbf{Y}_{S_r\bigcup S})\leq \varepsilon.
\end{equation}
Define the following distribution using 
 $f_{S_r,[n],g}$ and $\alpha$.
$$
{\cal D}^{f_{S_r,[n],g}}_\alpha\stackrel{def}{=}\left\{
\begin{array}{l}
\mathbf{y}\leftarrow \mathbf{Y}^\alpha\\
\mathrm{Output}\ \mathsf{same}^{*},\mbox{ if }\Delta g^\alpha (\mathbf{y})=0^n; \bot,\mbox{ otherwise}.\\
\end{array}
\right.
$$

In order to show that the real tampering experiment 
$$\mathrm{Tamper}_{\mathbf{m},\alpha}^{f_{S_r,[n],g}}=\mathsf{AMDdec}(\mathsf{AMDenc}(\mathbf{m})\oplus \mathsf{WTdec}(\Delta g^\alpha (\mathbf{X}^\alpha)))$$
is close to its simulation $\mbox{Patch}({\cal D}^{f_{S_r,[n],g}}_\alpha,\mathbf{m})$, 
we use 
 the  intermediate variable 
$$\mathbf{T}'=\mathsf{AMDdec}(\mathsf{AMDenc}(\mathbf{m})\oplus \mathsf{WTdec}(\Delta g^\alpha (\mathbf{Y}^\alpha))).$$ Now,
$$
\begin{array}{l}
\mbox{SD}(\mathrm{Tamper}_{\mathbf{m},\alpha}^{f_{S_r,[n],g}};\mbox{Patch}({\cal D}^{f_{S_r,[n],g}}_\alpha,\mathbf{m}))\\
\leq \mbox{SD}(\mathrm{Tamper}_{\mathbf{m},\alpha}^{f_{S_r,[n],g}};\mathbf{T}')+\mbox{SD}(\mathbf{T}';\mbox{Patch}({\cal D}^{f_{S_r,[n],g}}_\alpha,\mathbf{m}))\\
\stackrel{(\mathrm{i})} \leq \mbox{SD}(\mathbf{X}^\alpha_S;\mathbf{Y}^\alpha_S)+\mbox{SD}(\mathbf{T}';\mbox{Patch}({\cal D}^{f_{S_r,[n],g}}_\alpha,\mathbf{m}))\\
\stackrel{(\mathrm{ii})} \leq \varepsilon+\mbox{SD}(\mathbf{T}';\mbox{Patch}({\cal D}^{f_{S_r,[n],g}}_\alpha,\mathbf{m}))\\
\stackrel{(\mathrm{iii})}\leq \varepsilon+\delta,
\end{array}
$$
where  inequality (i) follows from the fact that $\mathrm{Tamper}_{\mathbf{m},\alpha}^{f_{S_r,[n],g}}$ and $\mathbf{T}'$ are only different at $\mathbf{X}^\alpha$ and $\mathbf{Y}^\alpha$,   inequality  (ii) follows from (\ref{eq: ssr}) and inequality  (iii) follows from the fact that $\mathbf{T}'$ and $\mbox{Patch}({\cal D}^{f_{S_r,[n],g}}_\alpha,\mathbf{m})$ have different values only when $\Delta g^\alpha (\mathbf{Y}^\alpha)\neq 0^n$ and $\mathbf{T}'\neq \bot$, which happens with probability at most $\delta$ according to the $\delta$-tamper detection security of the AMD code.

\item ``More than half of $\bar{S}_r$  are overwrite functions'' ($n^{ow}_{\bar{S}_r}> \frac{|\bar{S}_r|}{2}$): 

In this case, let $S$ be the index set of the non-overwrite components of $g^\alpha$ in $\bar{S}_r$. 
From the assumption, 
we have $|S|<\frac{|\bar{S}_r|}{2}$ and (\ref{eq: ssr}) holds. 
Let ${\cal D}^{f_{S_r,[n],g}}_\alpha$ be the distribution of the random variable $\mathsf{Dec}(g^\alpha(\mathbf{Y}^\alpha))$. 
Note that $\mathrm{Tamper}_{\mathbf{m},\alpha}^f=\mathsf{Dec}(g^\alpha(\mathbf{X}^\alpha))$. We need to show these two random variables are close.
$$
\mbox{SD}(\mathsf{Dec}(g^\alpha(\mathbf{X}^\alpha));\mathsf{Dec}(g^\alpha(\mathbf{Y}^\alpha)))\leq \mbox{SD}(\mathbf{X}^\alpha_S;\mathbf{Y}^\alpha_S)\leq\varepsilon,
$$
where the first inequality follows because $(g^\alpha(\mathbf{X}^\alpha))_{\bar{S}}=(g^\alpha(\mathbf{Y}^\alpha))_{\bar{S}}$ and the second inequality follows from (\ref{eq: ssr}).
\end{enumerate}

To bound the exact security of the NM-code, we define an intermediate distribution ${\cal D'}^{f_{S_r,[n],g}}$ that (unlike ${\cal D}^{f_{S_r,[n],g}}$) depends on message $\mathbf{m}$.
$$
\mathsf{Pr}\left[{\cal D'}^{f_{S_r,[n],g}}=\gamma\right]\nonumber= \sum_{\alpha\in\{0,1\}^{n\rho_r}} \mathsf{Pr}[\mathbf{X}_{S_r}=\alpha]\cdot\mathsf{Pr}[{\cal D}^{f_{S_r,[n],g}}_\alpha=\gamma].
$$
Let $\mathbf{\tilde{M}'}=\mbox{Patch}({\cal D'}^{f_{S_r,[n],g}},\mathbf{m})$ and $\mathbf{\tilde{M}}=\mbox{Patch}({\cal D}^{f_{S_r,[n],g}},\mathbf{m})$. We compute
$$
\begin{array}{ll}
\mbox{SD}(\mathrm{Tamper}_{\mathbf{m}}^{f_{S_r,[n],g}};\mathbf{\tilde{M}})
&\leq \mbox{SD}(\mathrm{Tamper}_{\mathbf{m}}^{f_{S_r,[n],g}};\mathbf{\tilde{M}'})+\mbox{SD}(\mathbf{\tilde{M}'};\mathbf{\tilde{M}})\\
&\stackrel{(\mathrm{i})}{\leq}(\varepsilon+\delta)+\mbox{SD}(\mathbf{\tilde{M}'};\mathbf{\tilde{M}})\\
&\stackrel{(\mathrm{ii})}{\leq}(\varepsilon+\delta)+\varepsilon=2\varepsilon+\delta,
\end{array}
$$
where (i) follows because $\mathrm{Tamper}_{\mathbf{m}}^{f_{S_r,[n],g}}$ is written as expected value (over $\alpha$) of $\mathrm{Tamper}_{\mathbf{m},\alpha}^{f_{S_r,[n],g}}$ according to (\ref{sum}) and for each $\alpha$ it is shown above that $\mbox{SD}(\mathrm{Tamper}_{\mathbf{m},\alpha}^{f_{S_r,[n],g}};\mbox{Patch}({\cal D}^{f_{S_r,[n],g}}_\alpha,\mathbf{m}))\leq \varepsilon+\delta$; (ii) follows because $\mathbf{\tilde{M}'}$ and $\mathbf{\tilde{M}}$ are defined in the same way (see ${\cal D'}^{f_{S_r,[n],g}}$ and ${\cal D}^{f_{S_r,[n],g}}$) with different distributions $\mathbf{X}_{S_r}$ and $\mathbf{Y}_{S_r}$ which are $\varepsilon$ close according to privacy of wiretap II.

The rate of the coding scheme is $\frac{k}{n}<\frac{\ell}{n}$, which according to Lemma \ref{lem: wtupperbound}
is upper bounded by $1-\rho$.
\qed
\end{proof}

\begin{lemma}\label{lem: small rhow'}
When $\rho_w<\frac{1-\rho_r}{2}$, the coding scheme $(\mathsf{Enc}, \mathsf{Dec})$ in Theorem \ref{con: cute bit'} is a $(\rho_r,\rho_w)_{\mathrm{BIT}}$-$\mathrm{NMC}$ with exact security 
$2\varepsilon+\delta$,
if $\rho\geq\rho_r+\rho_w$. 
\end{lemma}

\begin{proof}
This is a special case of Theorem \ref{con: cute bit'} when $\rho_w<\frac{1-\rho_r}{2}$.  In this case 
 the number of overwrite components can not exceed $\frac{|\bar{S}_r|}{2}$.  Following the proof steps of the theorem, the case 
 $n^{ow}_{\bar{S}_r}> \frac{|\bar{S}_r|}{2}$  in the proof  will not occur and one only has to make sure non-malleability is provided for $n^{ow}_{\bar{S}_r}$ values in the range $[0,n\rho_w]$. Let $S$ be the index set of the non-overwrite components of $\Delta g^\alpha$ (or equivalently the overwrite components of $g^\alpha$) in $\bar{S}_r$. We have $|S|\leq n\rho_w$ and hence $|S_r\bigcup S|=|S_r|+|S|\leq n(\rho_r+\rho_w)$. Now (\ref{eq: ssr}) will hold as long as the $(\rho,\varepsilon)$-$\mathrm{WT}$ satisfies $\rho\geq\rho_r+\rho_w$. 
The rest is identical  to the proof above (using only Case 1).
\qed
\end{proof}

\remove{
\begin{proof}[of Theorem \ref{con: cute bit'}]
The first claim follows from Lemma \ref{lem: rhow=1'}. Note that a $(\rho_r,1)_{\mathrm{BIT}}$-$\mathrm{NMC}$ is a $(\rho_r,\rho_w)_{\mathrm{BIT}}$-$\mathrm{NMC}$ for any $\rho_w$.
The second claim follows from Lemma \ref{lem: small rhow'}.
\qed
\end{proof}
}

It has been proved  \cite{WtIIcapacity} that the capacity of  binary $(\rho,\varepsilon)$-$\mathrm{WT}$ codes with   indistinguishability security 
is 
$1-\rho$.  It is however unknown if linear $(\rho,\varepsilon)$-$\mathrm{WT}$ codes  can achieve this rate.
\remove{
We do not know whether the capacity of binary linear $(\rho,\varepsilon)$-$\mathrm{WT}$ codes is still $1-\rho$ with the additional linearity property. 
}
The rate of  the resulting  NM-code is at most $1-\rho$, which is less than $1-\rho_r$ the achievable rate of $(\rho_r,\rho_w)_{\mathrm{BIT}}$-$\mathrm{NMC}$.
\remove{
Even if binary linear $(\rho,\varepsilon)$-$\mathrm{WT}$ codes at rate $1-\rho$ exists, this construction still falls short of rate: $1-\rho<1-\rho_r$. 
}
We leave explicit construction of capacity-achieving $(\rho_r,\rho_w)_{\mathrm{BIT}}$-$\mathrm{NMC}$ as an open question. 

\smallskip
\noindent
\textbf{Constructions of linear wiretap II codes.} 
The construction in Theorem \ref{con: cute bit'} requires a WtII code with leakage parameter $\rho=\frac{1+\rho_r}{2}\geq \frac{1}{2}$. 
The following explicit construction (coset coding)  gives binary linear $(\rho,0)$-$\mathrm{WT}$ codes. 

\begin{lemma}[\cite{WtII}]\label{lem: WtII} Let $G_{(n-k)\times n}$ be a generator matrix of an $[n,n-k,d]$-code $\mathcal{C}$ with dual distance $d^\bot$. Append $k$ rows to $G$ such that the obtained matrix $\left [\begin{array}{c} G\\\hat{G}\end{array} \right ]$ is of full rank.  
Define the encoder $\mathsf{WTenc}:\mathbb{F}_q^k\rightarrow\mathbb{F}_q^n$ as follows.
$$
\mathsf{WTenc}(\mathbf{m})=[\mathbf{R}\ \mathbf{m}]\left [\begin{array}{c} G\\\hat{G}\end{array} \right ],\mbox{ where }\mathbf{R}\stackrel{\$}{\leftarrow}\mathbb{F}_q^{n-k}.
$$
The message set $\mathbb{F}_q^k$ is in one-to-one correspondence with the cosets $\mathbb{F}_q^n\slash\mathcal{C}$ of $\mathcal{C}$ in the space $\mathbb{F}_q^n$. The decoder $\mathsf{WTdec}$ uses a parity-check matrix $H$ of the code $\mathcal{C}$ to efficiently identify the coset of the received word and output the corresponding message. Then $(\mathsf{WTenc}, \mathsf{WTdec})$ is a linear $\left(\frac{d^\bot-1}{n},0\right)$-$\mathrm{WT}$ code.
\end{lemma}

Binary linear codes with minimum distance $d> n/2$ exist; 
see for example 
 \cite{Markus}.  Instantiating $C$ in Lemma \ref{lem: WtII} with the dual of such codes result in 
  binary linear $(\rho,0)$-$\mathrm{WT}$ codes with $\rho\geq \frac{1}{2}$.

\vspace{1mm}
\noindent\textbf{An explicit family of $(0,1)_{\mathrm{BIT}}$-$\mathrm{NMC}$.}
For a $[2^h-1,2^h-1-h,3]$-Hamming code, 
the dual code 
is a $[2^h-1,h,2^{h-1}]$-Simplex code and has 
 $d^\bot=2^{h-1}$.  The WtII code can tolerate $\rho=\frac{1}{2}$ with $\varepsilon=0$. 
Using the AMD code 
in Lemma \ref{lem: AMD} with this code gives a bit-wise independent non-malleable code of length $2^h-1$.
 This is an explicit  construction of  
 non-malleable codes for the function family $\bt$. 
\remove{
For exa $h=6$, 
Theorem \ref{con: cute bit'} gives a bit-wise independent non-malleable code of length $63$ that encodes $2$ bits message and achieves exact security $\frac{1}{2}$. If we only encode one bit message, the length can be further reduced to $31$ (let $h=5$) with the same exact security. 
The exact security can be made arbitrarily small as the codeword length grows.
}
This  
code  can be made linear time encoding/decoding if the AMD construction is replaced with the linear time AMD construction in \cite{linear time}.

\section{Applications to Communication Security} \label{sec: applications}
Our motivation for introducing the function family $\rbtl$  is to model physical layer adversaries. In the following we give two applications of the NM-codes for $\rbtl$ family in widely studied communication settings.  In both cases we only consider one round protocols.
\subsection{Wiretap Channel II with Active Adversary}
Wiretap II model with active adversary was first studied in \cite{Lai Lifeng}, where the
eavesdropped and tampered components were restricted to the
same set.
In  the model proposed in  \cite{AWTP}  
 the adversary can read a fraction $\rho_r$, and add noise to a fraction $\rho_w$, 
and the goal is to provide secrecy (indistinguishability) and correct message recovery.
It was proved that  the rate upper bound for these codes is 
$1-\rho_r-\rho_w$, and so 
when $\rho_r+\rho_w>1$, one needs to relax privacy or reliability requirements.
\remove{
 has to consider privacy and a weaker form of reliability than message recovery. A model called limited view algebraic tampering or LVAT considered such an adversary \cite{ICITS}. The LVAT adversary can read a fraction $\rho$, and add noise to all of the codeword components. This is equivalent to a ($\rho,1$)-AWTP adversary. The wiretap II model with active LVAT adversary has two security goals: wiretap II privacy against $\rho$ fraction reading and detection of change of message \cite{ICITS}. These two goals can be achieved simultaneously for any value of $\rho<1$ and with rate $1-\rho$. 

Now w
}
We consider a wiretap II model where  the active adversary  can tamper with the codeword using functions in $\rbtlr$.
\remove{
who has even more power than the LVAT adversary. We give the adversary the extra power of overwriting any amount of codeword components and, similar to all previous wiretap II models with active adversary, the choice of the positions/values to overwrite is made after observing the content in the reading set. Note that for this adversary, detection of change of message is impossible. This is because the adversary can completely overwrite the received codeword to another valid codeword, for which the decoder will not be able to detect. 
}
\begin{definition}\label{def_awtpchannel}
A $(\rho_r,1)$-active adversary wiretap II channel is  
a communication channel between Alice and Bob that is (partially) controlled by an adversary Eve with two capabilities: Read and Write.

\begin{itemize}
\item Read:  Eve  selects a fraction $\rho_r$ of the components of the codeword to read. 

\item Write: Eve uses the read components to add errors to, or write over,  possibly all 
components of the codeword. 

 This  is equivalent to applying a function in $\rbtlr$ to the codeword.
\end{itemize}
\end{definition}

Codes for this channel must provide security (indistinguishability) and non-malleability.
\begin{definition}\label{def_awtpcode}
A $(\rho_r,1)$-active adversary wiretap II code 
is a coding scheme $(\mathsf{Enc},\mathsf{Dec})$ that 
guarantees the following two security properties.
\begin{itemize}
\item {\em Secrecy:} For any pair of messages $\mathbf{m}_0$ and $\mathbf{m}_1$, any reading set $S_r\subset [n]$ of size $|S_r|\leq n\rho_r$,
$$
\mathsf{SD}(\mathsf{Enc}(\mathbf{m}_0)_{S_r};\mathsf{Enc}(\mathbf{m}_1)_{S_r})\leq \varepsilon.
$$

\item {\em Non-malleability:} 
$(\mathsf{Enc},\mathsf{Dec})$ is non-malleable with respect to $\rbtlr$.
\end{itemize}
\end{definition}

Capacity of a  $(\rho_r,1)$-active adversary wiretap II channel is the highest achievable rate of coding schemes for this channel.
%
Using the results in Section \ref{sec: bounds} we can prove the following theorem.

\begin{theorem}\label{th: active WtII}
The capacity of $(\rho_r,1)$-active adversary wiretap II code is $1-\rho_r$.
\end{theorem}

\begin{proof}
We first show that a strong $(\rho_r,1)_{\mathrm{BIT}}$-$\mathrm{NMC}$ is a $(\rho_r,1)$-active adversary wiretap II code. The secrecy property follows from Lemma \ref{lem: strong is privacy}. 
The non-malleability property is satisfied because a strongly non-malleable coding scheme is non-malleable.
The lower bound $1-\rho_r$ for strong $(\rho_r,1)_{\mathrm{BIT}}$-$\mathrm{NMC}$ then gives a lower bound for $(\rho_r,1)$-active adversary wiretap II code. 

The upper bound $1-\rho_r$ follows because a $(\rho_r,1)$-active adversary wiretap II code is a wiretap II code with leakage $\rho_r$ (passive adversary).
\qed
\end{proof}

\remove{  THIS IS A CONFUSING REMARK - WHY SHOULD THESE BE RELATED??
\begin{remark} We note that a $(\rho_r,1)$-active adversary wiretap II code is not necessarily a strong $(\rho_r,1)_{\mathrm{BIT}}$-$\mathrm{NMC}$. The codes in Theorem \ref{th: active WtII construction} are good examples.
\end{remark}
}

We have the following two explicit constructions of efficient $(\rho_r,1)$-active adversary wiretap II codes using our constructions of $(\rho_r,\rho_w)_{\mathrm{BIT}}$-$\mathrm{NMC}$.  
Constructing  capacity-achieving $(\rho_r,1)$-active adversary wiretap II codes is 
 an open question.
\begin{theorem}\label{th: active WtII construction}
The 
constructions in Theorem \ref{th: basic construction} with $t'>n\rho_r$ and $d'>\frac{n(1-\rho_r)}{4}$,  
and 
Theorem \ref{con: cute bit'} with $\rho=\frac{1+\rho_r}{2}$,
each gives  a $(\rho_r,1)$-active adversary wiretap II code.
\end{theorem}
\begin{proof}
The non-malleability property follows directly from Theorem \ref{th: basic construction} and Theorem \ref{con: cute bit'}, respectively. The secrecy property follows from the $t'$-uniform property of LECSS and the privacy of the $(\rho,\varepsilon)$-$\mathrm{WT}$ code due to the parameter choices $t'>n\rho_r$ and $\rho=\frac{1+\rho_r}{2}>\rho_r$, respectively. 
\qed
\end{proof}

\remove{
\begin{remark} Note that a $(\rho_r,1)$-active adversary wiretap II code is not necessarily a strong $(\rho_r,1)_{\mathrm{BIT}}$-$\mathrm{NMC}$. 
For example, the construction in Theorem \ref{con: cute bit'} does not achieve strong non-malleability. This can be seen by considering a tampering function that overwrites more than $\rho=\frac{1+\rho_r}{2}$ fraction of the codeword components to $\beta$ and keep the rest of the components unchanged. Note that when the original value of the codeword components are $\beta$, the codeword will stay the same. According to Definition \ref{def: strong non-malleability}, this probability will be assigned to the symbol $\mathsf{same}^*$ in the strong tampering experiment. But since the length of $\beta$ is bigger than $n\rho$, the probability that the overwritten components have value $\beta$ will be significantly deferent for at least a pair of different messages, implying different probability for the symbol $\mathsf{same}^*$. This contradicts strong non-malleability.
\end{remark}
}

\subsection{Secure Message Transmission in Networks} 
\label{sec: second application}

In the model of Secure Message Transmission (SMT) \cite{DDWY93}, Alice is connected to Bob  by a set of $n$ node-disjoint network paths (also called wires).
The adversary can adaptively choose  a   subset of paths 
to eavesdrop and arbitrarily modify.
Although the original model considered adversaries who can select  possibly distinct sets of wires for listening, corrupting, and blocking,  
 SMT problem has been mainly studied  for  $(t,n)$-threshold adversaries  
who adaptively select  $t$ out of $n$ wires and  arbitrarily modify   them.  In the following we only consider this model.

A {\em  1-round $(\varepsilon,\delta)$-SMT protocol } is a   coding scheme with a pair  of  algorithms ($\mathsf{SMTenc}, \mathsf{SMTdec}$):  
 $\mathsf{SMTenc}$ encodes a message  ${\bf m_S}$ in $\mathcal{M}$ to a codeword (also called {\em protocol transcript})  $\mathbf{c}= (\mathbf{c_1},\cdots \mathbf{c_n}) \in (\mathbb{F}_q)^n$ where $\mathbf{c_i}$ is sent over wire $i$ (referred to as wire $i$ transcript), and $ \mathsf{SMTdec}$ decodes the received transcripts 
 to  ${\bf m_R}$ in   $\mathcal{M} \cup \{ \perp\}$, 
guaranteeing  privacy loss (indistinguishability security)  is at most $\varepsilon$,  and probability of error ($\Pr[{\bf m_S} \neq {\bf m_R}]$)  is bounded by $\delta$.
\remove{
\begin{definition} \label{smt}
A $(t,n)$-threshold  $(\varepsilon,\delta)$-secure message transmission (SMT)  protocol 
over $n$ wires, $t$ of which are corrupted by the adversary,  
satisfies the following properties for any choice of $t$ wires by the adversary: 
\begin{itemize}
\item {\em Secrecy: }For any pair of messages $\mathbf{m}_0,\mathbf{m}_1\in \cM$ and for all  adversary  $A$,
\[
\max_{\mathbf{m}_0, \mathbf{m}_1} SD(\mathsf{View}_A(\mathsf{SMTenc}(\mathbf{m}_0),r_A), 
\mathsf{View}_A(\mathsf{SMTenc}(\mathbf{m}_1),r_A))\leq \varepsilon_{SMT}.
\]
\item {\em Reliability:} The probability of  the receiver outputting an  incorrect   message is bounded by,
\[
\Pr[M_S \neq M_R]\leq \delta_{SMT},  
\]
where the message distribution is chosen by the adversary.
\end{itemize}
\end{definition}
}

It has been proved \cite{Franklin} that  $(\varepsilon,\delta)$-SMT protocols  exist   only if $n\geq 2t+1$,
 and this is irrespective of the number of protocol rounds.
%
In the following our goal is to show that  one can remove this restriction if the reliability goal is reduced to non-malleability. 

We consider adversaries that tamper with the protocol transcript according to functions in a {\em tampering function family} 
defined below.
We consider an SMT adversary with the following capabilities:
 the adversary  (i)  controls   $t=n\rho_r$ wires and (ii)  tampers obliviously with all other wires by either (algebraically) adding an offset, or 
 setting  the value to a 
  new value (overwrite).
  Compared to the traditional threshold $(t,n)$ adversary, this new adversary has the extra capability of tampering with all the wires also.
We consider the following set of symbol-wise independent Add and Overwrite (AO) tampering functions.
$$
\mathcal{F}_{\mathrm{AO}}^{[n]}=\left(\mathcal{F}^{add}\cup \mathcal{F}^{ow}\right)^{n},
$$
where $\mathcal{F}^{add}=\{f_\Delta(\mathbf{x})=\mathbf{x}+\Delta|\Delta\in\mathbb{F}_q\}$ denotes the set of additive tampering functions and $\mathcal{F}^{ow}=\{f_\mathbf{c}(\mathbf{x})=\mathbf{c}|\mathbf{c}\in\mathbb{F}_q\}$ denotes the set of overwrite tampering functions.

\vspace{1mm}
\noindent
{\bf Relation with $\bt$.}
Generalisation of $\bt$ from binary to $q$-ary alphabet $\mathbb{F}_q$, is symbol-wise independent tampering family $\mathcal{F}_{\mathrm{SIT}}^{[n]}\stackrel{def}{=}\left(\mathcal{F}^{all}\right)^{n}$ where $\mathcal{F}^{all}$ denotes all functions from $\mathbb{F}_q$ to $\mathbb{F}_q$. 
The class $\mathcal{F}_{\mathrm{AO}}^{[n]}$  is a subset of $\mathcal{F}_{\mathrm{SIT}}^{[n]}$ and  has size $|\mathcal{F}_{\mathrm{AO}}^{[n]}|=(2q)^n$.
 This is much smaller than $\mathcal{F}_{\mathrm{SIT}}^{[n]}$ that is of size $|\mathcal{F}_{\mathrm{SIT}}^{[n]}|=(q^q)^n$. It is only in the special case 
for  $q=2$, we have $\mathcal{F}_{\mathrm{AO}}^{[n]}=\bt$.

\begin{definition}\label{def: SMT tampering functions} 
\begin{equation}\label{eq: SMT set}
\AO    \stackrel{def}=     \left\{f_{S_r,[n],g}\ | \ S_r\in {\cal S}^{[n]}_{\rho_r} ,  
g: \mathbb{F}_q^{n\rho_r}\rightarrow   \mathcal{F}_{\mathrm{AO}}^{[n]}
 \right\},
\end{equation}
where the tampering function $f_{S_r,[n],g}:\mathbb{F}_q^n\rightarrow\mathbb{F}_q^n$ is given as follows:
for $\mathbf{x}\in\mathbb{F}_q^n$, depending on the value of $\mathbf{x}_{S_r}$,   $g$ selects a  
symbol-wise tampering function  $g( \mathbf{x}_{S_r} )$ from   $\mathcal{F}_{\mathrm{AO}}^{[n]}$  
which   is applied to $\mathbf{x}$.
That is,  
\begin{equation}\label{eq: SMT function}
f_{S_r,[n],g}(\mathbf{x}) = g( \mathbf{x}_{S_r})(\mathbf{x}).    
\end{equation} 
To simplify notations, we let $ g^{\mathbf{x}_{S_r}} \stackrel{def}=   g( {\mathbf{x}}_{S_r} )$ and write $g^{\mathbf{x}_{S_r}}=(g^{\mathbf{x}_{S_r}}_1,\cdots,g^{\mathbf{x}_{S_r}}_n)$, where $g^{\mathbf{x}_{S_r}}_i$ is either an additive function or an overwrite function, for $i=1,\cdots,n$. 
\end{definition}

\begin{definition}\label{def: nm SMT} 
A ($1$-round) $(n,\f,\varepsilon,\delta)$-non-malleable secure message transmission or $(n, \f,\varepsilon,\delta)$-$\mathrm{NM}$-$\mathrm{SMT}$ is  a protocol  over $n$ wires, defined by a pair  of  algorithms ($\mathsf{SMTenc}, \mathsf{SMTdec}$),
that for an adversary with access to corruption strategies in $\f$, the following properties are satisfied.
\remove{
  with messages in $\mathcal{M}$ and  codewords in $ (\mathbb{F}_q)^n$, 
,  and 
 the adversary ha
satisfies the following properties for any choice of $t$ wires by the adversary:
}
\begin{itemize}
\item {\em Secrecy: }For any pair of messages $\mathbf{m}_0,\mathbf{m}_1\in \cM$ and for any  adversary strategy $A$ embodied   by $\f$, 
\[
\max_{\mathbf{m}_0, \mathbf{m}_1} SD(\mathsf{View}_A(\mathsf{SMTenc}(\mathbf{m}_0)); 
\mathsf{View}_A(\mathsf{SMTenc}(\mathbf{m}_1)))\leq \varepsilon,
\]
where $\mathsf{View}_A(\cdot)$ is a random variable representing leakage.
\item Non-malleability: ($\mathsf{SMTenc}, \mathsf{SMTdec}$) is non-malleable with respect to $\f$ and with exact security $\delta$.
\end{itemize}
\end{definition}


\begin{theorem} \label{th: SMT}
The construction in Theorem \ref{con: cute bit'} with a $(\frac{1+\rho_r}{2},\varepsilon)$-$\mathrm{WT}$ code 
over the alphabet $\mathbb{F}_q$ and an AMD code with $\delta$-tamper detection security gives a $(n,\AO,\varepsilon,2\varepsilon+\delta)$-$\mathrm{NM}$-$\mathrm{SMT}$.
\end{theorem}
\begin{proof}

The construction in Theorem \ref{con: cute bit'} with the above parameter setting gives a $q$-ary NM-code with respect to $\AO$. 
The proof  relies on the properties of  difference function that will hold for functions in $\mathcal{F}_{\mathrm{AO}}^{[n]}$ only (and not  $\mathcal{F}_{\mathrm{SIT}}^{[n]}$).
Secrecy follows from the  indistinguishability security of $q$-ary wiretap II code.
\qed
\end{proof}

Linear $(\frac{1+\rho_r}{2},\varepsilon)$-$\mathrm{WT}$ codes over large alphabet can be constructed using the coset code construction in Lemma \ref{lem: WtII} using for example a Maximum Distance Separable (MDS) code. This gives explicit $(\frac{1+\rho_r}{2},0)$-$\mathrm{WT}$ codes  with rate $1-\frac{1+\rho_r}{2}$.






\section{Conclusion}
We proposed a   family of bitwise tampering functions that were motivated by physical layer adversaries  and  were specified by a pair of parameters $(\rho_r, \rho_w)$, defining the eavesdropping and tampering capabilities of the adversary.  Allowing the adversary to select tampering based on the eavesdropped information models powerful adversaries and results in a class of functions that is much larger (double exponential) than the widely studied independent bit tampering class.
We defined non-malleable code with respect to this class and  proved a number of rate bounds that  fully characterize capacity of  codes  that provide  strong non-malleability, as well as   capacity of  (default) non-malleable codes when $\rho_r\leq \rho_w$.
 We  also gave two modular constructions, with the second one using wiretap II codes that, using coset code construction of these codes, effectively constructs  NM-codes from linear error correcting codes and AMD codes.
 
 There are many open questions and directions for future research.
 We left tight upper bound and capacity of (default) NM-codes  when $\rho_w<\rho_r$,  as an open problem.
Also none of our construction are capacity achieving, and so construction of capacity achieving codes  remains open.

Our main results are for binary codes.  However in SMT setting, transcripts of wires are $q$-ary values.  Extending the results for $\rbtl$ to $q$-ary case, strengthens our current NM-SMT construction for function class $\AO$, and   allows more powerful adversaries in network setting be  tolerated.
Our explicit  construction for NM-SMT uses $q$-ary wiretap II codes. One can also  adapt the construction of $q$-ary codes in \cite{linear time} to obtain secure NM-SMT. We leave  this for future work. 

Non-malleability was  originally motivated for providing tamper resilience in tamperable storages.  Our work is the first to consider application in well motivated communication settings of wiretap II codes with active adversaries, and secure communication in networks.
Other applications of non-malleability in communication scenarios, including modelling collusion attacks, are interesting directions for future work.
The function class $\rbtl$ assumes tampering on the components of $S_w$  are independent of each other, and depends on the read values over $S_r$, only. A more general case is when tampering of bits in $S_w$ are correlated.

{
\appendix
\section*{Appendices}
\addcontentsline{toc}{section}{Appendices}
\renewcommand{\thesubsection}{\Alph{subsection}}

\subsection{Appendix to the proof of Theorem \ref{th: basic construction}}\label{apdx: proof of basic construction}

\begin{enumerate}
\item $n^{ow}_{\bar{S}_r}\in \left[0,t'-n\rho_r\right]$.   We consider the effect of the difference function defined by  (\ref{eq: difference}) on
 codewords in $C_\alpha$.
 According to  above, the number of non-overwrite components of $\Delta g^\alpha$ in $\bar{S}_r$ will be at most $t'-n\rho_r$, and according to the $t'$-uniform property of LECSS,
 these components are each uniformly distributed,  and are jointly $(t'-n\rho_r)$-wise independent. 
  This means that the  non-overwrite components of $\Delta g^\alpha(\mathbf{X})$ in $\bar{S}_r$ are uniformly distributed over $\{0,1\}^{n^{ow}_{\bar{S}_r}}$. 
  %
 The  rest of the components of 
  $\Delta g^\alpha(\mathbf{X})$   in $\bar{S}_r$  are overwrite components (correspond to non-overwrite components of $g^\alpha$).
  Thus the distribution  $\Delta g^\alpha(\mathbf{X})$ is independent of the input (an AMD codeword $\mathsf{AMDenc}(\mathbf{m})$) of LECSS. 
  The third step of tampering experiment is applying the decoding function Dec on the tampered codeword (of the NM-code). 
Using $g^\alpha(\mathbf{x})=\mathbf{x}\oplus \Delta g^\alpha(\mathbf{x})$, we have
$$
\begin{array}{ll}
\mathsf{Dec}\left (g^\alpha(\mathbf{X})\right )
&=\mathsf{Dec}\left (\mathbf{X}\oplus\Delta g^\alpha(\mathbf{X})\right )\\
&=\mathsf{AMDdec}\left (\mathsf{AMDenc}(\mathbf{m})\oplus\mathsf{LECSSdec}(\Delta g^\alpha(\mathbf{X}))\right ),\\
\end{array}
$$
where the second equality follows from the linearity of the LECSS. 

To find the distribution of the output of the tampering experiment in this case,  denoted by $\left (\mathrm{Tamper}_{\mathbf{m}}^{f_{S_r,[n],g}}|\mathbf{X}_{S_r}=\alpha\right )$, we note that: 
\begin{itemize}
\item Conditioned on $\mathsf{LECSSdec}(\Delta g^\alpha(\mathbf{X}))= 0^{\ell}$,  the output of Dec is going to be $\mathsf{AMDdec}\left (\mathsf{AMDenc}(\mathbf{m})\right)=\mathbf{m}$.
\item Consider $\mathsf{LECSSdec}(\Delta g^\alpha(\mathbf{X}))\neq 0^{\ell}$. Since the distribution of $\Delta g^\alpha(\mathbf{X})$
 is completely determined by $f_{S_r,[n],g}$ and $\alpha$, it  
 is independent of the randomness of the AMD code. 
 According to Lemma \ref{lem: AMD},  the probability that the AMD decoder not outputting $\bot$ is at most $\mathsf{Pr}[\mathsf{LECSSdec}(\Delta g^\alpha(\mathbf{X}))\neq 0^{\ell}]\cdot\delta\leq \delta$.
\end{itemize}
Thus we can define the distribution $\mathcal{D}^{f_{S_r,[n],g}}_\alpha$ as follows:

$$
   \mathcal{D}^{f_{S_r,[n],g}}_\alpha \stackrel{def}= \left\{ 
    \begin{array}{l}
 \mathbf{z} \leftarrow  \left(\Delta g^\alpha(\mathbf{X})|\mathbf{X}_{S_r}=\alpha\right)
    \\
 \mathrm{Output} \ \mathsf{same}^*, \mbox{ if  LECSSdec}(\mathbf{z})=0^{\ell}; 
 \mathrm{Output}\ \bot,  \mbox{ otherwise}.
    \end{array}
    \right.
$$


This distribution  will be different from 
 the tampering experiment 
  when the AMD decoder  fails to  output $\bot$ for $\mathbf{z}\notin\{0^{\ell},\bot\}$. We then have
$$
\left (\mathrm{Tamper}_{\mathbf{m}}^{f_{S_r,[n],g}}|\mathbf{X}_{S_r}=\alpha\right )\stackrel{\delta}{\approx}\mbox{Patch}(\mathcal{D}^{f_{S_r,[n],g}}_\alpha, \mathbf{m}). 
$$

\item $n^{ow}_{\bar{S}_r}\in \left(t'-n\rho_r,\frac{|\bar{S}_r|}{2}\right]$. 
 The distribution $\mathcal{D}^{f_{S_r,[n],g}}_\alpha$ will only output $\perp$, using LECSS decoder error detection property.
 Firstly using the LECSS linearity, the decoder output will depend on the result of $\Delta g^\alpha$ on a codeword  (in $C_\alpha$).
 The given number of $n^{ow}_{\bar{S}_r}$ translates into the same number of non-overwrite for $\Delta g^\alpha$ on components in ${\bar{S}_r}$, and 
the rest of  components being overwrite  function.
The codeword components of $C_\alpha$  in $\bar{S}_r$  are  $(t'-n\rho_r)$-uniform,  and 
as said earlier non-overwrite functions do not affect a column probability distribution, which  are uniform because $C_\alpha$  in $\bar{S}_r$  is  $(t'-n\rho_r)$-uniform.

If none of the vectors in the list $\mathsf{Array}(\Delta g^\alpha(\mathbf{x})|\mathbf{x}\in C_\alpha))$ correspond to a valid codeword of LECSS, the LECSS decoder output will be always $\perp$. 
If there is a vector $\omega$ in 
$\mathsf{Array}(\Delta g^\alpha(\mathbf{x})|\mathbf{x}\in C_\alpha)$ that corresponds to a LECSS codeword,  there will be an undetected error. Note that $\omega$ may appear more than once in the list. Next, the distance property of LECSS (together with $(t'-n\rho_r)$-uniform property) is utilised to claim that the density of valid LECSS codewords in $\mathsf{Array}(\Delta g^\alpha(\mathbf{x})|\mathbf{x}\in C_\alpha)$ is very small. 
The argument had been used  in \cite[Proof of Theorem 4.1,  Case  3]{DzPiWi} for the function class $\bt$,  to quantify the decoder error.
  Authors showed that for a $(t,d)$-$\mathrm{LECSS}$ if the minimum distance  is $d > n/4$, 
the error probability is given by,
 \begin{equation}
\Pr[\mathsf{LECSSdec}(\Delta) \neq \perp )]\leq \frac{1}{2^t} +\left(\frac{t}{n(d/n-1/4)^2}\right)^{t/2},
  \end{equation}
where $\Delta$ is a vector random variable of $n$ components, more than half but less than $n-t$ of which are fixed values and the rest of components are $t$-uniform.
      
We  use the same argument and  make the following adjustments.  Firstly, the tampering functions are applied to $C_\alpha$ and so the 
tampered words 
$\mathsf{Array}(\Delta g^\alpha(\mathbf{x})|\mathbf{x}\in C_\alpha)$ will have fixed values on index set $S_r$.
So the part of components that can be different (between two vectors in the list) are in $\bar{S}_r$. For our proof  we consider $\bar{S}_r$.
Thus we only need $d' > (n-n\rho_r)/4$.
Also the non-overwrite components of  $\mathsf{Array}(\Delta g^\alpha(\mathbf{x})|\mathbf{x}\in C_\alpha)$ in $\bar{S}_r$  are  $(t'-\rho_r n)$-uniform and so we have,

  \begin{equation}
\mathsf{Pr}\left[\mathsf{LECSSdec}\left(\Delta g^\alpha(\mathbf{X})\right )\neq \bot|\mathbf{X}_{S_r}=\alpha\right] \leq \frac{1}{2^{t'-n\rho_r}} +\left(\frac{t'-n\rho_r}{n(\frac{d' }{n}-\frac{1-\rho_r}{4})^2}\right)^{\frac{t'-n\rho_r}{2}}.
  \end{equation}

\item $n^{ow}_{\bar{S}_r}\in \left(\frac{|\bar{S}_r|}{2}, n-t'\right)$. 
 The distribution $\mathcal{D}^{f_{S_r,[n],g}}_\alpha$ will only output $\perp$, using LECSS decoder error detection property.
We study $\mathsf{Array}(g^\alpha(\mathbf{x})|\mathbf{x}\in C_\alpha)$, the list of tampered codewords, and bound the probability of LECSS decoder cannot detect the error.
The argument  is similar to above.
This corresponds to Case 4 in the proof of Theorem 4.1 in \cite{DzPiWi}.   Using the required adjustment as outlined above, we will have

  \begin{equation}
\mathsf{Pr}\left[\mathsf{LECSSdec}\left(g^\alpha(\mathbf{X})\right )\neq \bot|\mathbf{X}_{S_r}=\alpha\right] \leq \frac{1}{2^{t'-n\rho_r}} +\left(\frac{t'-n\rho_r}{n(\frac{d' }{n}-\frac{1-\rho_r}{4})^2}\right)^{\frac{t'-n\rho_r}{2}}.
  \end{equation}

(It is worth noting that the argument in the above two cases use two different sets ( $\mathsf{Array}(\Delta g^\alpha(\mathbf{x})|\mathbf{x}\in C_\alpha)$ and $\mathsf{Array}(g^\alpha(\mathbf{x})|\mathbf{x}\in C_\alpha)$) that have the property that overwrite components in one, corresponds to non-overwrite component in the other.  The choice of the list is to allow many overwrite components and minimise the $d'$-distance requirement.)

\item $n^{ow}_{\bar{S}_r}\in \left[n-t',|\bar{S}_r|\right]$. 

This is the case that most of the codeword is overwritten, and non-overwritten part  is uniformly distributed.
This is because less than $t'$ components are not overwritten, and the set of vectors $\widetilde{C_\alpha}$ is $t'-n\rho_r$-uniform.
Thus the distribution is independent of $m$, and the decoder output distribution will have the same property also.
The distribution $\mathcal{D}^{f_{S_r,[n],g}}_\alpha$ in this case is defined as follows:

$$
   \mathcal{D}^{f_{S_r,[n],g}}_\alpha \stackrel{def}= \left\{ 
    \begin{array}{l}
 \mathbf{z} \leftarrow    g^\alpha(\mathbf{Y})
    \\
 \mathrm{Output}\  \mathsf{Dec}(\mathbf{z})      \end{array}
    \right.
$$

Since the simulation and the tampering experiment are identical in this case, 
$$
\left (\mathrm{Tamper}_{\mathbf{m}}^{f_{S_r,[n],g}}|\mathbf{X}_{S_r}=\alpha\right )\equiv\mbox{Patch}(\mathcal{D}^{f_{S_r,[n],g}}_\alpha, \mathbf{m}).
$$

\end{enumerate}


\begin{thebibliography}{10} 

\bibitem{Wyner}
A. D. Wyner. ``The wire-tap channel''. Bell System Technical Journal,
54:pp. 1355-1367, 1975.

\bibitem{Lmoreext}
Aggelos Kiayias, Feng-Hao Liu and 	Yiannis Tselekounis, Practical Non-Malleable Codes from l-more Extractable Hash Functions,
Proceedings of the 2016 ACM SIGSAC Conference on Computer and Communications Security
Pages 1317-1328.

\bibitem{BTV}
Bellare, Mihir, Stefano Tessaro, and Alexander Vardy. ``Semantic
security for the wiretap channel.'' Advances in Cryptology. CRYPTO
2012. Springer Berlin Heidelberg, 2012. 294-311.

\bibitem{LDNMC}
Dana Dachman-Soled, Feng-Hao Liu, Elaine Shi, Hong-Sheng Zhou. Locally Decodable and Updatable Non-malleable Codes and Their Applications. In: Dodis Y., Nielsen J.B. (eds) Theory of Cryptography. TCC 2015. Lecture Notes in Computer Science, vol 9014. pp 427-450. Springer, Berlin, Heidelberg.

\bibitem{DDWY93}
Danny Dolev, Cynthia Dwork, Orli Waarts, and Moti Yung. 1993. Perfectly secure message transmission. J. ACM 40, 1 (January 1993), 17-47. 

\bibitem{optimalcomputational}
Divesh Aggarwal, Shashank Agrawal, Divya Gupta, Hemanta K. Maji, Omkant Pandey and Manoj Prabhakaran, Optimal Computational Split-state Non-malleable Codes. In: Kushilevitz E., Malkin T. (eds) Theory of Cryptography 2016. Lecture Notes in Computer Science, vol 9563. pp 393-417. Springer, Berlin, Heidelberg.

\bibitem{NMreduction} 
Divesh Aggarwal, Yevgeniy Dodis, Tomasz Kazana, and Maciej Obremski. Non-malleable reductions and applications. In Rocco A. Servedio and Ronitt Rubinfeld, editors, 47th ACM STOC, pages 459-468, Portland, OR, USA, June 14-17, 2015. ACM Press.

\bibitem{additivecombinatorics} 
Divesh Aggarwal, Yevgeniy Dodis, and Shachar Lovett. Non-malleable codes from additive combinatorics. In David B. Shmoys, editor, 46th ACM STOC, pages 774-783, New York, NY, USA, May 31-June 3, 2014. ACM Press.

\bibitem{NMcryptography}
Danny Dolev, Cynthia Dwork, and Moni Naor. Nonmalleable cryptography. SIAM J. Comput., 30(2):391-437, 2000.


\bibitem{small-depth circuits}
Eshan Chattopadhyay and Xin Li. Non-malleable Codes and Extractors for Small-Depth Circuits, and Affine Functions. To appear STOC 2017.






\bibitem{10split}
Eshan Chattopadhyay and David Zuckerman. Non-malleable codes against constant split-state tampering. In 55th FOCS, pages 306-315, Philadelphia, PA, USA, October 18-21, 2014. IEEE Computer Society Press.

\bibitem{ECC}
F.J. MacWilliams and N.J.A. Sloane, The Theory of Error-Correcting Codes, North- Holland, 1977.

\bibitem{Liucomputational}
Feng-Hao Liu and Anna Lysyanskaya. Tamper and leakage resilience in the split-state model. In Reihaneh Safavi-Naini and Ran Canetti, editors, CRYPTO 2012, volume 7417 of LNCS, pages 517-532, Santa Barbara, CA, USA, August 19-23, 2012. Springer, Heidelberg, Germany.

\bibitem{ICITS}
Fuchun Lin, Reihaneh Safavi-Naini, Pengwei Wang. Detecting Algebraic Manipulation in Leaky Storage Systems. In: Nascimento A., Barreto P. (eds) Information Theoretic Security. ICITS 2016. Lecture Notes in Computer Science, vol 10015. pp 129-150.




\bibitem{ChenHao}
Hao Chen, Ronald Cramer, Shafi Goldwasser, Robbert de Haan, Vinod Vaikuntanathan. Secure computation from random error correcting codes. In Moni Naor, editor, Advances in Cryptology  EUROCRYPT 2007, volume 4515 of Lecture Notes in Computer Science, pages 291-310. Springer-Verlag, Berlin, Germany, May 2007.

\bibitem{ChGu0} 
Mahdi Cheraghchi and Venkatesan Guruswami. Capacity of non-malleable codes. In Moni Naor, editor, ITCS 2014, pages 155-168, Princeton, NJ, USA, January 12-14, 2014. ACM.

\bibitem{ChGu1}
Mahdi Cheraghchi and Venkatesan Guruswami. Non-malleable coding against bit-wise and split-state tampering. In Yehuda Lindell, editor, TCC 2014, volume 8349 of LNCS, pages 440-464, San Diego, CA, USA, February 24-26, 2014. Springer, Heidelberg, Germany.

\bibitem{circuit paper}
Marshall Ball, Dana Dachman-Soled, Mukul Kulkarni, Tal Malkin
Non-malleable Codes for Bounded Depth, Bounded Fan-In Circuits
Advances in Cryptology  EUROCRYPT 2016
Volume 9666 of the series Lecture Notes in Computer Science pp 881-908.

\bibitem{Markus}
Markus Grassl. Bounds on the minimum distance of linear codes and quantum codes. Online
available at \url{http://www.codetables.de/}, 2007. Accessed on 2012-07-23.

\bibitem{Franklin}
Matthew Franklin and Rebecca N. Wright. Secure Communication in Minimal Connectivity
Models. Journal of Cryptology, January 2000, Volume 13, Issue 1, pp 9-30.

\bibitem{Matthieu}
Matthieu Bloch and Joao Barros (2011). Physical-layer security : from information theory to security engineering. Cambridge University Press, Cambridge.

\bibitem{ITLDNMC}
Nishanth Chandran, Bhavana Kanukurthi, Srinivasan Raghuraman Information-Theoretic Local Non-malleable Codes and Their Applications. In: Kushilevitz E., Malkin T. (eds) Theory of Cryptography. Lecture Notes in Computer Science, vol 9563. pp 367-392. Springer, Berlin, Heidelberg.


\bibitem{block-wise}
Nishanth Chandran, Vipul Goyal, Pratyay Mukherjee, Omkant Pandey, and Jalaj Upadhyay,
Block-wise Non-malleable Codes, proceedings of the 43rd International Colloquium on Automata, Languages, and Programming-ICALP 2016.







\bibitem{WtII} Ozarow, L. H., and A. D. Wyner. "Wire-tap channel II." AT \& T Bell Laboratories Technical Journal 63.10(1984):2135-2157.

\bibitem{AWTP} 
Pengwei Wang and Reihaneh Safavi-Naini. A Model for Adversarial Wiretap Channels.
IEEE Transactions on Information Theory, vol. 62, no. 2, FEB 2016.



\bibitem{AMD} Ronald Cramer, Yevgeniy Dodis, Serge Fehr, Carles Padro, and daniel Wichs. Detection of Algebraic Manipulation with Applications to Robust Secret Sharing and Fuzzy Extractors. Advances in Cryptology-EUROCRYPT 2008, pages 471-488, 2008.




\bibitem{linear time} 
Ronald Cramer, Ivan Damgard, Nico Dottling, Irene Giacomelli and Chaoping Xing. Linear-Time Non-Malleable Codes in the Bit-Wise Independent Tampering Model, http://eprint.iacr.org/2016/397

\bibitem{public key} 
Sandro Coretti, Ueli Maurer, Bjorn Tackmann and Daniele Venturi. From Single-Bit to Multi-bit Public-Key Encryption via Non-malleable Codes. In: Dodis Y., Nielsen J.B. (eds) Theory of Cryptography. TCC 2015. Lecture Notes in Computer Science, vol 9014. pp 532-560. Springer, Berlin, Heidelberg.

\bibitem{CNMC}
Sebastian Faust, Pratyay Mukherjee, Jesper Buus Nielsen, and Daniele Venturi. Continuous non- malleable codes. In Yehuda Lindell, editor, TCC 2014, volume 8349 of LNCS, pages 465-488, San Diego, CA, USA, February 24-26, 2014. Springer, Heidelberg, Germany.

\bibitem{efficientNMC}
Sebastian Faust, Pratyay Mukherjee, Daniele Venturi, and Daniel Wichs. Efficient non-malleable codes and key-derivation for poly-size tampering circuits. In Phong Q. Nguyen and Elisabeth Oswald, editors, EUROCRYPT 2014, volume 8441 of LNCS, pages 111-128, Copenhagen, Denmark, May 11-15, 2014. Springer, Heidelberg, Germany.



\bibitem{Majicrypto} 
Shashank Agrawal, Divya Gupta, Hemanta K. Maji, Omkant Pandey, and Manoj Prabhakaran. Explicit non-malleable codes against bit-wise tampering and permutations. In Rosario Gennaro and Matthew Robshaw, editors, Advances in Cryptology - CRYPTO 2015 - 35th Annual Cryptology Conference, Santa Barbara, CA, USA, August 16-20, 2015, Proceedings, Part I, volume 9215 of Lecture Notes in Computer Science, pages 538-557. Springer, 2015.

\bibitem{MajiTCC}
Shashank Agrawal, Divya Gupta, Hemanta K. Maji, Omkant Pandey, and Manoj Prabhakaran. A rate- optimizing compiler for non-malleable codes against bit-wise tampering and permutations. In Yevgeniy Dodis and Jesper Buus Nielsen, editors, TCC 2015, Part I, volume 9014 of LNCS, pages 375-397, Warsaw, Poland, March 23-25, 2015. Springer, Heidelberg, Germany.



\bibitem{onebit}
Stefan Dziembowski, Tomasz Kazana, and Maciej Obremski. Non-malleable codes from two-source extractors. In Ran Canetti and Juan A. Garay, editors, CRYPTO 2013, Part II, volume 8043 of LNCS, pages 239-257, Santa Barbara, CA, USA, August 18-22, 2013. Springer, Heidelberg, Germany.

\bibitem{DzPiWi}
Stefan Dziembowski, Krzysztof Pietrzak, and Daniel Wichs. Non-malleable codes. In Andrew Chi-Chih Yao, editor, ICS 2010, pages 434-452, Tsinghua University, Beijing, China, January 5-7, 2010. Tsinghua University Press.


\bibitem{Lai Lifeng} 
V. Aggarwal, Lifeng Lai, A.R. Calderbanand H.V. Poor. Wiretap channel type II with an active eavesdropper, IEEE International Symposium on Information Theory (ISIT) 2009, pp. 1944-1948. 

\bibitem{TF}
Y. Desmedt. Major Security Problems with the ``Unforgeable'' (Feige-)Fiat-Shamir Proofs of Identity and How to Overcome Them. In Congress on Computer and Communication Security and Protection Securicom 88, Paris, France, pp. 147-159, SEDEP Paris France, 1988.

\bibitem{TDC}
Zahra Jafargholi and Daniel Wichs. Tamper Detection and Continuous Non-malleable Codes. In: Dodis Y., Nielsen J.B. (eds) Theory of Cryptography. TCC 2015. Lecture Notes in Computer Science, vol 9014. pp 451-480. Springer, Berlin, Heidelberg.

\bibitem{WtIIcapacity}
Ziv Goldfeld,  Paul Cuff,  and Haim H. Permuter. ``Semantic-Security Capacity for Wiretap Channels of Type II'' Information Theory, IEEE Transactions on VOL. 62, NO. 7, JULY 2016.






\end{thebibliography}
\end{document}